\documentclass[12pt]{article}

\usepackage[utf8]{inputenc}
\usepackage[T1]{fontenc}
\usepackage{lmodern}

\usepackage[a4paper,margin=1in]{geometry}

\usepackage{xcolor} 
\usepackage{hyperref}

\hypersetup{%
  colorlinks=true,
  linkbordercolor=red,
  citecolor=green!70!black,
}

\usepackage{authblk} 
\usepackage{titlesec} 
 
\usepackage[linesnumbered,algoruled,boxed,lined]{algorithm2e}
\usepackage{complexity} 
 
\usepackage{tikz}
\usetikzlibrary{positioning,backgrounds,patterns,calc, graphs, shapes, fit}
\usetikzlibrary{arrows.meta}

\usepackage{enumitem}

\usepackage[noadjust]{cite} 

\usepackage{amsmath}
\usepackage{amsthm}

\usepackage{thmtools}
\usepackage{thm-restate}

\usetikzlibrary{patterns}
\usepackage{hyperref}
 \hypersetup{%
  colorlinks=true,
  linkbordercolor=red,
  citecolor=green!70!black,
}
\usepackage{subcaption}
 \usepackage{graphicx}

\newcommand{\red}[1]{{\color{red}[#1]}}

\newcommand{\ETH}{{\textsf{ETH}}}

\newcommand{\tw}{\textsf{tw}}

\newcommand{\calO}{\mathcal{O}}

\newcommand{\true}{\texttt{True}}
\newcommand{\false}{\texttt{False}}

\newcommand{\yes}{\textsc{Yes}}

\newtheorem{lemma}{Lemma}[section]
\newtheorem{corollary*}{Corollary}
\newtheorem{corollary}{Corollary}
\newtheorem{definition}{Definition}[section]
\newtheorem{observation}{Observation}[section]
\newtheorem{proposition}{Proposition}[section]

\makeatletter
\newcommand{\customlabel}[2]{%
\protected@write \@auxout {}{\string \newlabel {#1}{{#2}{}}}}
\makeatother

\newcommand{\defproblem}[3]{
  \vspace{1mm}
\noindent\fbox{
  \begin{minipage}{0.96\textwidth}
  \begin{tabular*}{\textwidth}{@{\extracolsep{\fill}}lr} #1 \\ \end{tabular*}
  {\bf{Input:}} #2  \\
  {\bf{Question:}} #3
  \end{minipage}
  }
  \vspace{1mm}
}

\title{A Finer View of the Parameterized Landscape of Labeled Graph
Contractions}

\author[1]{Yashaswini Mathur}
\author[2]{Prafullkumar Tale\thanks{The author is supported by the ARG-MATRICS grant MaPLE (ARGM/2026/0089).}} 

\affil[1]{Indian Institute of Science Education and Research Bhopal, India}
\affil[2]{Indian Institute of Science Education and Research Pune, India}

\date{February 2026} 

\begin{document}

\maketitle

\begin{abstract}
We study the \textsc{Labeled Contractibility} problem, where the input consists of two vertex-labeled graphs $G$ and $H$, and the goal is to determine whether $H$ can be obtained from $G$ via a sequence of edge contractions.

Lafond and Marchand~[WADS 2025] initiated the parameterized complexity study of this problem, 
showing it to be \(\W[1]\)-hard when parameterized by the number 
\(k\) of allowed contractions. 
They also proved that the problem is fixed-parameter tractable when 
parameterized by the treewidth \(\tw\) of \(G\), via 
an application of Courcelle's theorem with a non-elementary dependence on the parameter.

In this work, we present a constructive 
fixed-parameter algorithm for \textsc{Labeled Contractibility} with running time 
\(2^{\mathcal{O}(\tw^2)} \cdot |V(G)|^{\mathcal{O}(1)}\).
We also prove that unless the Exponential Time Hypothesis (\ETH) fails, 
it does not admit
an algorithm running in time 
\(2^{o(\tw^2)} \cdot |V(G)|^{\mathcal{O}(1)}\).
This result adds \textsc{Labeled Contractibility} to a small list
of problems for which a \(2^{\Theta(\tw^2)}\) dependence is optimal under \ETH.

We further strengthen existing hardness results by showing that the problem remains \NP-hard even when both input graphs have bounded maximum degree. 
We also investigate parameterizations by \((k + \delta(G))\), where \(\delta(G)\) denotes the degeneracy of \(G\), and rule out the existence of subexponential-time algorithms. This answers a question on subexponential fixed-parameter tractability
raised by Lafond and Marchand~[WADS 2025]. 
We additionally provide an improved \FPT\ algorithm running in time
\((\delta(H)+1)^k \cdot |V(G)|^{\mathcal{O}(1)}\).
Finally, we analyze a brute-force algorithm for \textsc{Labeled Contractibility} with running 
time \(|V(H)|^{\mathcal{O}(|V(G)|)}\), and show that this running time is optimal under \ETH.

\vskip5pt\noindent{}{\em Keywords}: Labeled Contraction, ETH Lower Bound, Treewidth 

\vskip5pt\noindent{}\emph{Declaration About Usage of LLMs:}
All research ideas and technical contributions in this paper were developed by the authors. 
ChatGPT was used to improve the presentation and clarity of the manuscript.
\end{abstract}

\newpage

\section{Introduction}

Graph modification problems are a foundational topic in algorithmic graph theory. A general formulation of the $\mathcal{F}$-$\texttt{M}$-\textsc{Modification} problem asks: given a graph $G$ and a fixed set $\texttt{M}$ of allowed modification operations (such as vertex deletion, edge addition/deletion, or edge contraction), can $G$ be transformed into a graph in class $\mathcal{F}$ using at most $k$ operations from $\texttt{M}$?

It is well known that the $\mathcal{F}$-$\texttt{M}$-\textsc{Modification} problem is \NP-hard for many natural graph classes $\mathcal{F}$ when $\texttt{M}$ consists of vertex or edge deletions~\cite{10.1145/800133.804355,LewisY80}. Hardness results for edge contraction were first shown by Watanabe et al.~\cite{WatanabeAN81} and Asano and Hirata~\cite{AsanoH83}. It is noteworthy that while the problem becomes trivially solvable when $\texttt{M}$ includes only vertex or edge deletions and $\mathcal{F}$ is a singleton, this is not the case for edge contractions. In fact, Brouwer and Veldman~\cite{BrouwerV87} established that the problem remains \NP-complete even when $\mathcal{F}$ is a singleton containing a four-vertex graph, such as the path $P_4$ or the cycle $C_4$. This difference highlights the increased algorithmic complexity introduced by edge contractions compared to vertex or edge deletions.

Graph modification problems have played a central role in the development of parameterized complexity and the design of key algorithmic techniques. A representative list of relevant works includes~\cite{BliznetsFPP15,BliznetsFPP18,Cao16,Cao17,CaoM15,CaoM16,CrespelleDFG20,DrangeDLS22,DrangeFPV14,DrangeP18,FominKPPV14,FominV13}. For a comprehensive overview, the survey by Crespelle et al.~\cite{crespelle2020survey} provides an in-depth discussion of modification problems involving edge deletion, editing, and completion.

A series of recent papers have studied the parameterized complexity of $\mathcal{F}$ - \textsc{Contraction} for various graph classes $\mathcal{F}$, including paths and trees~\cite{heggernes2014contracting},
generalizations and restrictions of trees~\cite{agarwal2019parameterized, agrawal2017paths}, cactus graphs~\cite{krithika2018fpt}, 
bipartite graphs~\cite{guillemot2013faster, heggernes2013obtaining}, planar graphs~\cite{golovach2013obtaining}, grids~\cite{tale2020parameterized}, cliques~\cite{cai2013contracting}, 
split graphs~\cite{agrawal2019split}, chordal graphs~\cite{lokshtanov2013hardness}, bicliques~\cite{martin2015computational}, and degree-constrained graph classes~\cite{Belmonte:2014,golovach2013increasing, DBLP:conf/iwpec/0001T20}.

Recently, Lafond and Marchand~\cite{DBLP:journals/corr/abs-2502-16096} initiated the parameterized study of edge contraction problems in \emph{uniquely labeled graphs}, where each vertex has a different label. They introduced the following problem:

\defproblem{\textsc{Labeled Contractibility}}{Two vertex-labeled graphs $G$ and $H$ with $V(H) \subseteq V(G)$.}{Is $H$ a labeled contraction of $G$?}

\noindent See Definition~\ref{def:labeled-graphs} in Section~\ref{sec:prelims} for formal definitions.

They also introduced a generalization, motivated by applications in phylogenetic networks~\cite{marchand2024finding}, which are rooted acyclic directed graphs with labeled leaves representing evolutionary histories:

\defproblem{\textsc{Maximum Common Labeled Contraction}}{Two vertex-labeled graphs $G$ and $H$, and an integer $k$.}{Do there exist labeled contraction sequences $S_1$ and $S_2$ such that $G/S_1 = H/S_2$ and $|S_1| + |S_2| \le k$?}

This second problem is closely related to the first. As shown in~\cite{DBLP:journals/corr/abs-2502-16096}, $H$ is a labeled contraction of $G$ if and only if there exists a common contraction of both $G$ and $H$ to a labeled graph $H'$ with at least $|V(H)|$ vertices. This equivalence allows complexity results to be transferred between the two problems. 
We refer readers to the same work~\cite{DBLP:journals/corr/abs-2502-16096} for 
additional results and applications.

Lafond and Marchand~\cite{DBLP:journals/corr/abs-2502-16096} showed that \textsc{Labeled Contractibility} is \(\W[1]\)-hard when parameterized by the natural parameter \(k = |V(G)\setminus V(H)|\), the number of allowed edge contractions. This motivates the exploration of structural parameters. They established that the problem admits a fixed-parameter tractable (\FPT) algorithm when parameterized by the treewidth \(\tw\) of the input graph \(G\) by invoking a variation of Courcelle’s theorem~\cite[Theorem 11.73]{DBLP:series/txtcs/FlumG06}. However, this approach yields a non-elementary dependence on the parameter.

Our first contribution is a constructive dynamic programming algorithm with a running time that is singly exponential in \(\tw^2\). We complement this with a matching lower bound under the Exponential Time Hypothesis (\ETH):

\begin{restatable}{theorem}{treewidth}
\label{thm:treewith}
\textsc{Labeled Contractibility}
\begin{itemize}[nolistsep]
    \item admits an algorithm running in time \(2^{\mathcal{O}(\tw^2)} \cdot |V(G)|^{\mathcal{O}(1)}\); but
    \item does not admit an algorithm running in time \(2^{o(\tw^2)} \cdot |V(G)|^{\mathcal{O}(1)}\), unless the \ETH\ fails.
\end{itemize}
\end{restatable}

We note that this lower bound also holds when parameterized by the pathwidth, 
which is a larger parameter. 
Most known \FPT-algorithms that use dynamic programming on tree decompositions have running times, 
and matching lower bounds, of the form \(2^{\mathcal{O}(\tw)} \cdot n^{\mathcal{O}(1)}\)~\cite[Chapters 7,13]{cygan2015parameterized} or \(2^{\mathcal{O}(\tw \log \tw)} \cdot n^{\mathcal{O}(1)}\)~\cite{DBLP:journals/siamcomp/LokshtanovMS18, DBLP:journals/talg/CyganNPPRW22}. 
For a relatively rare set of problems that exhibit a double-exponential dependence on treewidth,
see~\cite{DBLP:conf/icalp/FoucaudGK0IST24, Tale24, ChakrabortyFM24} and references therein. 
To the best of our knowledge, the only problems with single-exponential-but-polynomial dependencies on pathwidth are mentioned in~\cite{DBLP:journals/iandc/CyganMPP17,DBLP:conf/mfcs/Pilipczuk11,DBLP:journals/iandc/SauS21}. Similar results parameterized by vertex cover number can be found in~\cite{agrawal2019split,DBLP:conf/stacs/FoucaudGK0IST25, BergougnouxKN22, ChakrabortyDFM25}.

Lafond and Marchand~\cite{DBLP:journals/corr/abs-2502-16096} observed that smaller parameters such as degeneracy do not yield tractable results. Specifically, it is known that the \textsc{Maximum Common Labeled Contraction} problem is \NP-hard even when both input graphs have constant maximum degree (and thus constant degeneracy), and that \textsc{Labeled Contractibility} is \NP-hard for graphs of bounded degeneracy~\cite[Theorem 8]{DBLP:conf/wabi/MarchandTSL24}. We strengthen this latter result by proving hardness under an even stricter constraint:

\begin{restatable}{theorem}{NPhardBoundedMaxDegree}
\label{thm:bounded-max-degree}
\textsc{Labeled Contractibility} is \NP-hard even when both \(G\) and \(H\) have bounded maximum degree.
\end{restatable}

These results suggest that tractability is unlikely when parameterizing solely by the solution size \(k\) or the degeneracy \(\delta(G)\). Nevertheless, Lafond and Marchand~\cite{DBLP:journals/corr/abs-2502-16096} showed the problem becomes tractable for the combined parameter \((k + \delta(G))\), yielding an algorithm with a running time of \((\delta(G) + 2k)^k \cdot n^{\mathcal{O}(1)}\). They asked whether a subexponential-time algorithm exists for this parameterization, a question we resolve negatively.

Hardness reductions for closely related common-contraction problems have also used \textsc{Set Splitting}, equivalently a positive variant of \textsc{Not-All-Equal-SAT}~\cite{DBLP:conf/wabi/MarchandTSL24}. For the conditional lower bound below, we instead use \textsc{1-in-3-SAT}, for which we establish the required \ETH\ lower bound before giving a linear-size reduction to \textsc{Labeled Contractibility}.

\begin{restatable}{theorem}{degeneracyeth}
\label{thm:bounded-degeneracy-eth}
\textsc{Labeled Contractibility} does not admit an algorithm running in time \(2^{o(|V(G)| + |E(G)|)}\) even when both \(G\) and \(H\) have bounded degeneracy, unless the \ETH\ fails.
\end{restatable}

This theorem rules out the possibility of a subexponential-time algorithm parameterized by \((k + \delta(G))\) under the assumption of the \ETH. We also provide the following algorithm, which improves upon previous work:

\begin{restatable}{theorem}{degeneracy}
\label{thm:degeneracy-algo}
\textsc{Labeled Contractibility} can be solved in time \((\delta(H) + 1)^k \cdot |V(G)|^{\mathcal{O}(1)}\), where \(k=|V(G)\setminus V(H)|\) is the number of required contractions and \(\delta(H)\) is the degeneracy of \(H\).
\end{restatable}

This algorithm improves upon the \((\delta(H) + k)^k \cdot n^{\mathcal{O}(1)}\) algorithm from~\cite[Section 3.3]{DBLP:journals/corr/abs-2502-16096}.

Finally, as our last contribution, we analyze a brute-force algorithm for \textsc{Labeled Contractibility} and establish its optimality under the \ETH. Notably, our analysis also applies to the more general \textsc{Maximum Common Labeled Contraction} problem.

\begin{restatable}{theorem}{Bruteforce}
\label{thm:brute-force-optimal}
\textsc{Labeled Contractibility}
\begin{itemize}[nolistsep]
    \item admits an algorithm running in time \(|V(H)|^{\mathcal{O}(|V(G)|)}\); but
    \item does not admit an algorithm running in time \(|V(H)|^{o(|V(G)|)}\), unless the \ETH\ fails.
\end{itemize}
\end{restatable}

\noindent \textbf{Organization.}
We adopt standard notation from graph theory and parameterized complexity, as outlined in Section~\ref{sec:prelims}.
In Section~\ref{sec:treewidth}, we present a dynamic programming algorithm based on tree decompositions and establish a matching lower bound. Specifically, we provide a parameter-preserving reduction from the \textsc{Sub-Cubic Partitioned Vertex Cover} problem~\cite{agrawal2019split}, which proves Theorem~\ref{thm:treewith}.
Section~\ref{sec:np-hard-max-degree-bound} proves Theorem~\ref{thm:bounded-max-degree} by a reduction from \textsc{Positive-Not-All-Equal-SAT}, showing \NP-hardness even for graphs of bounded maximum degree.
Section~\ref{sec:degeneracy} is dedicated to a conditional lower bound under the \ETH. Using a reduction from \textsc{1-in-3-SAT}, we prove this result even when both input graphs have bounded degeneracy. We also present an improved \FPT\ algorithm for the combined parameter \( (k+\delta(G)) \). These results together establish Theorems~\ref{thm:bounded-degeneracy-eth} and~\ref{thm:degeneracy-algo}.
Section~\ref{sec:brute-force} analyzes a brute-force algorithm and establishes its optimality under the \ETH. This is achieved via a reduction from the \textsc{Cross Matching} problem~\cite{DBLP:journals/toct/FominLMSZ21} and establishes Theorem~\ref{thm:brute-force-optimal}. Finally, Section~\ref{sec:conclusion} concludes with open problems and conjectures.
\section{Preliminaries}
\label{sec:prelims}

For a positive integer $n$, we denote $[n] = \{1, 2, \ldots, n\}$. Let $G$ be a graph. We denote by $V(G)$ and $E(G)$ the vertex set and edge set of $G$, respectively.
Two vertices $u, v \in V(G)$ are said to be \emph{adjacent} if $(u, v) \in E(G)$.
For a vertex $u \in V(G)$, let $N_G(u)$ denote the (open) neighborhood of $u$, i.e.,
the set of all vertices adjacent to $u$ in $G$, and define the 
closed neighborhood as $N_G[u] = N_G(u) \cup \{u\}$. 
We omit the subscript if the graph is clear from the context.
Given two disjoint subsets $X, Y \subseteq V(G)$, we say that $X$ and $Y$
are \emph{adjacent} if there exists an edge in $G$ with one endpoint 
in $X$ and the other in $Y$. 
For any $X \subseteq V(G)$, the subgraph of $G$ induced by $X$ is denoted by $G[X]$.
We write $G - X$ to denote the induced subgraph $G[V(G) \setminus X]$.
The \emph{maximum degree} of $G$, denoted $\Delta(G)$, is defined as 
$\max_{u \in V(G)} |N_G(u)|$. 
The \emph{degeneracy} of $G$, denoted $\delta(G)$, is the smallest integer $d$ 
such that every subgraph of $G$ contains a vertex of degree at most $d$.

To cope with the hardness of \NP-hard problems, these problems are studied
in the framework of \emph{parameterized complexity}. 
This approach enables a finer analysis of 
computational problems by 
introducing a parameter $\ell$ 
(e.g., solution-size, treewidth, vertex cover number), which may be significantly smaller than the input size. 
A problem is \emph{fixed-parameter tractable} 
(\FPT) if it admits an algorithm running in time $f(\ell) \cdot |I|^{\mathcal{O}(1)}$ 
for some computable function $f$ and instance size $|I|$. Problems that are hard for 
the class $\W[1]$ are believed not to admit such algorithms under standard assumptions, and those that 
remain \NP-hard even for fixed parameter values are called $\textsf{para}$-\NP-hard.
We refer readers to the book on parameterized algorithms~\cite{cygan2015parameterized}
for other related terms.

The Exponential Time Hypothesis (\ETH) of 
Impagliazzo and Paturi~\cite{DBLP:journals/jcss/ImpagliazzoP01} implies that 
the \textsc{3-Sat} problem on $n$ 
variables cannot be solved in time 
$2^{o(n)}$, ignoring polynomial factors.
The Sparsification Lemma of 
Impagliazzo et al.~\cite{DBLP:journals/jcss/ImpagliazzoPZ01} 
implies that if the 
\ETH\ holds, then there is no algorithm solving 
a \textsc{3-Sat} formula with $n$ variables 
and $m$ clauses in time $2^{o(n+m)}$, ignoring polynomial factors.

Throughout this work, we consider uniquely labeled graphs. Following the convention in~\cite{DBLP:journals/corr/abs-2502-16096}, we refer to vertices directly by their labels rather than using explicit labeling functions. Two labeled graphs $G$ and $H$ are considered \emph{equal}, denoted $G=H$, if and only if $V(G)=V(H)$ and $E(G)=E(H)$. This is distinct from the standard notion of graph isomorphism, where a bijection $\sigma: V(G) \to V(H)$ exists such that $(u, v) \in E(G)$ if and only if $(\sigma(u), \sigma(v)) \in E(H)$. In our setting, we assume that the labels of $H$ are a subset of those of $G$.

\begin{definition}
\label{def:labeled-graphs}
Let $G$ be a labeled graph and $(u,v) \in E(G)$. The \emph{labeled contraction} of the edge $(u,v)$ is an operation that transforms $G$ by:
\begin{itemize}[nolistsep]
    \item For every vertex $w \in N(v) \setminus N[u]$, adding the edge $(u,w)$ to $G$.
    \item Removing the vertex $v$ and all incident edges.
\end{itemize}
The resulting graph is denoted by $G/(u,v)$.
\end{definition}

The labeled contraction operation is illustrated on a toy example in Figure~\ref{fig:contraction-example}.

\begin{figure}
\centering
\begin{tikzpicture}[
    vertex/.style={circle, draw, inner sep=1.5pt, font=\small},
    rep/.style={circle, draw, double, inner sep=1.5pt, font=\small},
    witness/.style={rounded corners, fill, opacity=0.35},
    scale=1
]

\begin{scope}
\node at (1.5,2.3) {$G$};

\node[vertex]    (g1) at (0,2) {1};
\node[rep] (g2) at (-1,1) {2};
\node[vertex] (g5) at (-0.5,0) {5};

\node[rep]    (g3) at (1.5,1) {3};
\node[vertex] (g4) at (1,0) {4};

\draw (g1)--(g2);
\draw (g2)--(g5);
\draw (g5)--(g3);
\draw (g3)--(g4);
\draw[red, thick] (g2)--(g3);

\end{scope}

\draw[->, thick] (3,1) -- (4,1);

\begin{scope}[xshift=6cm]
\node at (1.5,2.3) {$M$};

\node[vertex] (m1) at (0.5,2) {1};
\node[rep] (m2) at (-0.5,0.8) {2};
\node[vertex] (m5) at (0.5,0) {5};
\node[vertex] (m4) at (2,0) {4};

\draw (m1)--(m2);
\draw (m2)--(m4);
\draw[purple, thick] (m2)--(m4);
\draw (m2)--(m5);
\end{scope}

\end{tikzpicture}
\caption{Contraction along edge $(2,3)$}
\label{fig:contraction-example}
\end{figure}

Note that in general, $G/(u,v) \neq G/(v,u)$, as the vertex that is retained after the contraction (and thus the label) is different. The two resulting graphs are, however, isomorphic under the standard unlabeled notion. For simplicity, in the remainder of this paper, we assume all contractions are labeled and refer to them simply as contractions.

A labeled contraction sequence $S$ on a graph $G$ is a sequence of vertex pairs $S=((u_1, v_1), \ldots, (u_k, v_k))$ such that for each $i \in \{0, \ldots, k-1\}$, if $G^i$ is the graph obtained after the first $i$ contractions (with $G^0 := G$), then $(u_{i+1}, v_{i+1}) \in E(G^i)$ and $G^{i+1} = G^i/(u_{i+1}, v_{i+1})$. The graph obtained after the full sequence is denoted by $G/S$. If $S$ is not a valid contraction sequence on $G$, then $G/S$ is undefined. We say that a graph $H$ is a labeled contraction of $G$ if there exists a contraction sequence $S$ such that $G/S=H$. A graph $M$ is a common labeled contraction of graphs $G$ and $H$ if it is a labeled contraction of both. A maximum common labeled contraction of $G$ and $H$ is a common labeled contraction with the largest possible number of vertices.

Consider two graphs $H$ and $G$ such that $V(H) \subseteq V(G)$.

\begin{definition}[\cite{DBLP:journals/corr/abs-2502-16096}]
\label{def:witness}
A witness structure of $G$ into $H$ is a partition $\mathcal{W}=\{W_1, \ldots, W_{|V(H)|}\}$ of $V(G)$ into non-empty sets satisfying:
\begin{itemize}[nolistsep]
    \item For each $W_i \in \mathcal{W}$, the induced subgraph $G[W_i]$ is connected.
    \item Each $W_i$ contains exactly one vertex from $V(H)$, called its representative vertex.
    \item For any distinct $u,v \in V(H)$, the edge $(u,v) \in E(H)$ if and only if the sets in $\mathcal{W}$ containing $u$ and $v$ are adjacent in $G$.
\end{itemize}
The sets in $\mathcal{W}$ are called \emph{witness sets}.
\end{definition}

We illustrate a witness structure in Figure~\ref{fig:Witness-Structure}.

\begin{figure}
\centering
\begin{tikzpicture}[
    vertex/.style={circle, draw, inner sep=1.5pt, font=\small},
    rep/.style={circle, draw, double, inner sep=1.5pt, font=\small},
    witness/.style={rounded corners, fill, opacity=0.35},
    scale=1
]

\begin{scope}
\node at (1.5,2.3) {$G$};

\node[rep]    (g1) at (0,2) {1};
\node[rep] (g2) at (-1,1) {2};
\node[vertex] (g5) at (-0.5,0) {5};

\node[rep]    (g3) at (1.5,1) {3};
\node[vertex] (g4) at (1,0) {4};

\node[vertex] (g6) at (2.5,0) {6};
\node[rep] (g7) at (3.5,1) {7};

\draw (g1)--(g2);
\draw (g2)--(g5);
\draw (g5)--(g3);
\draw (g3)--(g4);
\draw (g4)--(g6);
\draw (g6)--(g7);
\draw (g3)--(g7);
\draw (g4)--(g5);

\begin{pgfonlayer}{background}
    \node[witness, fill=red, fit=(g1)] {};  
    \node[witness, fill=orange, fit=(g2)(g5)] {};
    \node[witness, fill=cyan, fit=(g3)(g4)(g6)] {};
    \node[witness, fill=green, fit=(g7)] {};
\end{pgfonlayer}
\end{scope}

\draw[->, thick] (4.5,1) -- (5.3,1);

\begin{scope}[xshift=7cm]
\node at (1.2,2.3) {$H$};

\node[rep] (m1) at (0,2) {1};
\node[rep] (m2) at (-0.5,0.8) {2};
\node[rep] (m3) at (1,0.8) {3};
\node[rep] (m7) at (2.2,0.8) {7};

\draw (m1)--(m2);
\draw (m2)--(m3);
\draw (m3)--(m7);

\begin{pgfonlayer}{background}
    \node[witness, fill=red, fit=(m1)] {};  
    \node[witness, fill=orange, fit=(m2)] {};
    \node[witness, fill=cyan, fit=(m3)] {};
    \node[witness, fill=green, fit=(m7)] {};
\end{pgfonlayer}

\end{scope}

\end{tikzpicture}
\caption{Shaded regions in $G$ form the witness sets, each containing exactly one representative vertex (circled). Contracting each witness set yields $H$.}
\label{fig:Witness-Structure}
\end{figure}

There is a natural bijection between the sets in $\mathcal{W}$ and the vertices of $H$ via their representatives. It is known that $H$ is a labeled contraction of $G$ if and only if a witness structure of $G$ into $H$ exists (see Observation 2 in~\cite{DBLP:journals/corr/abs-2502-16096}). Furthermore, contractions within a single witness set can be performed in any order as long as the representative vertex is preserved (Observation 3 in the same work). For any vertex $u \in V(G)$, we denote by $\mathcal{W}(u) \in \mathcal{W}$ the unique set that contains $u$.
\section{Parameterized by Treewidth}
\label{sec:treewidth}

In this section, we present the proof of Theorem~\ref{thm:treewith} which we restate.

\treewidth*

In the first subsection, we describe the algorithmic result by presenting a dynamic 
programming algorithm on tree decomposition. 
In the second subsection, we establish a matching conditional lower bound.

\subsection{Algorithmic Result}

In this subsection we present a dynamic programming (DP) based algorithm for the
\textsc{Labeled Contraction} problem when parameterized by treewidth of the input graph.
For the sake of completeness, we start with the standard definitions.

\begin{definition}[\cite{cygan2015parameterized}]
A \emph{tree decomposition} of a graph \( G \) is a pair
\( (\mathcal{T}, \{X_t\}_{t \in V(\mathcal{T})}) \), where
\(\mathcal{T}\) is a rooted tree and each node \(t \in V(\mathcal{T})\)
is associated with a set \(X_t \subseteq V(G)\), referred to as a \emph{bag},
satisfying the following conditions:
\begin{itemize}[nolistsep]
    \item For every vertex \(v \in V(G)\), the set
    \(\{t \in V(\mathcal{T}) \mid v \in X_t\}\) is nonempty and induces
    a connected subtree of \(\mathcal{T}\);
    \item For every edge \(\{u,v\} \in E(G)\), there exists a node
    \(t \in V(\mathcal{T})\) such that \(\{u,v\} \subseteq X_t\).
\end{itemize}
The \emph{width} of a tree decomposition is
\(\max_{t \in V(\mathcal{T})}|X_t|-1\).
The \emph{treewidth} of \(G\), denoted \(\tw(G)\), is the minimum width
over all valid tree decompositions of \(G\).
\end{definition}

\begin{definition}[\cite{cygan2015parameterized}]
A \emph{nice tree-decomposition} of a graph \(G\) is a rooted tree-decomposition
\((\mathcal{T},\{X_t\}_{t\in V(\mathcal{T})})\) where
\(\mathcal{T}\) is rooted at \(r\) with \(X_r=\emptyset\), and each node of
\(\mathcal{T}\) belongs to one of the following types:
\begin{itemize}[nolistsep]
\item A leaf node \(t\) is a leaf of \(\mathcal{T}\) with \(X_t=\emptyset\).
\item An introduce vertex node \(t\) has one child \(t_1\) such that
\(X_t=X_{t_1}\cup\{x\}\), for some \(x\notin X_{t_1}\).
\item A forget node \(t\) has one child \(t_1\) such that
\(X_t=X_{t_1}\setminus\{x\}\), where \(x\in X_{t_1}\).
\item An introduce edge node \(t\) has one child \(t'\) with
\(X_t=X_{t'}\), and is labeled with an edge
\((u,v)\in E(G)\) such that \(\{u,v\}\subseteq X_t\).
We say that edge \((u,v)\) is introduced at \(t\).
\item A join node \(t\) has two children \(t_1,t_2\) such that
\(X_t=X_{t_1}=X_{t_2}\).
\end{itemize}
\end{definition}

Without loss of generality, we assume that a nice tree-decomposition of width
\(\tw\) is provided. If not, one can construct an optimum tree decomposition
in time \(2^{\mathcal{O}(\tw^2)}n^{\mathcal{O}(1)}\)~\cite{KorhonenL23}.

Recall that we denote \(G\cup H\) as the graph with vertex set \(V(G)\) and
edge set \(E(G)\cup E(H)\). The following lemma enables us to work with tree
decompositions of \(G\cup H\) instead of just \(G\), which is crucial for our
formulation.

\begin{lemma}[\cite{DBLP:journals/corr/abs-2502-16096}, Lemma 13]
If \(H\) is a contraction of \(G\), then \(\tw(G\cup H)\leq 2\tw(G)\).
\end{lemma}

Henceforth, we fix a nice tree decomposition \(\mathcal{T}\) of \(G\cup H\)
of width \(\mathcal{O}(\tw(G))\). If such a decomposition does not exist,
Lemma~3.1 implies that the instance is a No-instance. We further assume that
every edge of \(G\) is introduced exactly once at an introduce-edge node.

For each node \(t\in V(\mathcal{T})\), let \(\mathcal{T}_t\) denote the
subtree of \(\mathcal{T}\) rooted at \(t\), and let
\(V_t=\bigcup_{s\in V(\mathcal{T}_t)}X_s\).
We define \(G_t\) to be the graph with vertex set \(V_t\) and edge set
consisting of precisely those edges of \(G\) whose introduce-edge node belongs
to \(\mathcal{T}_t\). We define \(H_t=H[V_t\cap V(H)]\).
Notice that \(G_t\) need not be the graph induced by \(V_t\). In particular,
at an introduce-vertex node, no edge incident with the newly introduced vertex
has yet been introduced.

\begin{figure}[t]
\centering

\begin{subfigure}[t]{0.6\textwidth}
\centering
\begin{tikzpicture}
    \colorlet{lightgray}{gray!20}
    \colorlet{darkgray}{gray!60}

    \path[top color=lightgray, bottom color=lightgray, opacity=1] (0.5,0) -- (7.5,0) -- (9,-3.5) -- (-1,-3.5) -- cycle;

    \draw[top color=white, bottom color=white, opacity=1, draw=black, dotted, thick] 
        (0.5,0) -- (7.5,0) -- (8.5,2.3) -- (-0.5,2.3) -- cycle;

    \node[ellipse, fill=darkgray, minimum width=8cm, minimum height=1.5cm, opacity=1, draw=black] (Xt_oval) at (4, 0) {};

    \node[circle, fill=black, inner sep=1.5pt, label=right:{\scriptsize $u_8$}] (n1) at (0.7,0) {};
    \node[circle, fill=black, inner sep=1.5pt, label=left:{\scriptsize $u_7$}] (n2) at (2,0) {};
    \node[circle, fill=black, inner sep=1.5pt, label=left:{\scriptsize $u_1$}] (n3) at (3.5,0.3) {};
    \node[fill=black, inner sep=2pt, label=left:{\scriptsize $h_2$}] (n4) at (5,0.3) {};
    \node[circle, fill=black, inner sep=1.5pt, label=left:{\scriptsize $u_2$}] (n5) at (3.5,-0.3) {};
    \node[circle, fill=black, inner sep=1.5pt, label=left:{\scriptsize $u_9$}] (n6) at (5,-0.3) {};
    \node[circle, fill=black, inner sep=1.5pt, label=left:{\scriptsize $u_5$}] (n7) at (6.5,0) {};

    \node[circle, fill=black, inner sep=1.5pt] (n8) at (0.8,-1.6) {};
    \node[circle, fill=black, inner sep=1.5pt, label=left:{\scriptsize $u_3$}] (n9) at (2,-1.5) {};
    \node[circle, fill=black, inner sep=1.5pt] (n10) at (1,-1) {};
    \node[fill=black, inner sep=2pt, label=left:{\scriptsize $h_1$}] (n11) at (1.5,-2) {};

    \node[circle, fill=black, inner sep=1.5pt] (n12) at (3,-1.4) {};
    \node[circle, fill=black, inner sep=1.5pt] (n13) at (4,-1.2) {};
    \node[circle, fill=black, inner sep=1.5pt] (n14) at (4.5,-1.2) {};
    \node[circle, fill=black, inner sep=1.5pt] (n15) at (3.5,-2) {};

    \node[circle, fill=black, inner sep=1.5pt, label=right:{\scriptsize $u_4$}] (n16) at (6,-1.3) {};
    \node[fill=black, inner sep=2pt, label=above:{\scriptsize $h_3$}] (n17) at (7,-2) {};

    \node[circle, fill=black, inner sep=1.5pt] (h1) at (3.5, 1.5) {};
    \node[circle, fill=black, inner sep=1.5pt, label=right:{\scriptsize $u_6$}] (h2) at (4.5, 1.5) {};

    \node[draw, rounded corners=15pt, thick, fit=(n1)(n2)(n8)(n9)(n10)(n11), inner sep=7pt, label=below:{\scriptsize $W_1$}] {};
    \node[draw, rounded corners=15pt, thick, fit=(n3)(n4)(n5)(n6)(n12)(n13)(n14)(n15)(h1)(h2), inner sep=7pt, label={[label distance=9pt]below:{\scriptsize $W_2$}}] {};
    \node[draw, rounded corners=15pt, thick, fit=(n7)(n16)(n17), inner sep=7pt, label=below:{\scriptsize $W_3$}] {};

    \draw (h1) -- (h2);
    \draw (n1) -- (n10);
    \draw (n10) -- (n8);
    \draw (n2) -- (n9);
    \draw (n10) -- (n11);
    \draw (n9) -- (n11);
    \draw (h1) -- (n3);
    \draw (h2) -- (n4);
    \draw (n3) -- (n5);
    \draw (n4) -- (n6);
    \draw (n5) -- (n12);
    \draw (n6) -- (n13);
    \draw (n6) -- (n14);
    \draw (n12) -- (n15);
    \draw (n7) -- (n16);
    \draw (n16) -- (n17);
    \draw (h2) -- (n7);
    \draw (n9) .. controls (4,-3) .. (n16);

    \node[left] at (8.8,0) {$X_t$};
    \node[right] at (7.8, -3) {$G_t$}; 
    \node[right] at (6, 1.8) {$Rest~of~G$}; 

\end{tikzpicture}
\caption{Illustration of $G_t$ and $X_t$}
\end{subfigure}
\hfill
\begin{subfigure}[t]{0.35\textwidth}
\centering
\begin{tikzpicture}[
    sq/.style={rectangle, draw, minimum size=2mm},
    every node/.style={font=\small}
]

\node[sq] (n1) at (0,0) {};
\node[sq] (n2) at (2,0) {};
\node[sq] (n3) at (4,0) {};

\node[below=2pt of n1] {$h_1$};
\node[below=2pt of n2] {$h_2$};
\node[below=2pt of n3] {$h_3$};

\draw (n2) -- (n3);
\draw[bend left=40] (n1) to (n3);

\end{tikzpicture}
\caption{Graph $H$}
\end{subfigure}

\caption{Depicting definitions relevant to the algorithm}
\label{fig:side-by-side}
\end{figure}

Assume that \(H\) is a contraction of \(G\), and let
\(\mathcal{W}=\{W_1,\dots,W_{|V(H)|}\}\) denote a hypothetical but fixed
witness structure for this contraction. Recall that we refer to the unique
vertex \(h\in W_i\cap V(H)\) as the \emph{representative vertex} of the
witness set \(W_i\).

\begin{definition}[Active Vertices]
\label{def:active-vertices}
Let \(t\in V(\mathcal{T})\) be a node with bag \(X_t\).
A representative vertex \(h\in V(H)\) is active at \(t\) if its corresponding
witness set \(W_h\in\mathcal{W}\) intersects \(X_t\), i.e.,
\(W_h\cap X_t\neq\emptyset\). The set \(A_t\subseteq V(H)\) is the collection
of all active vertices at \(t\).
\end{definition}

\begin{observation}
\label{obs:active-bound}
The number of active vertices \(|A_t|\) is upper bounded by the size of the bag,
\(|A_t|\leq |X_t|\), as each active vertex corresponds to a unique witness set
that intersects \(X_t\). Furthermore, due to the connectivity properties of the
witness sets, once a witness set ceases to intersect the current bag, it cannot
intersect any bag above \(t\).
\end{observation}

Figure~\ref{fig:side-by-side} illustrates the information visible at a node
\(t\). The shaded region represents the processed graph \(G_t\), while the oval
represents the current bag \(X_t\). The intersections \(W_i\cap X_t\) determine
how the vertices of the bag are distributed among the witness sets. We also
record which vertices belonging to the same witness set are already connected
in its processed portion.

\begin{definition}[Witness Impression]
\label{def:new-witness-impression}
For a given witness structure \(\mathcal{W}\), the \emph{witness impression}
on the bag \(X_t\) of node \(t\) is the triple
\((\mathcal{P},\mathcal{C},\tau)\), where:
\begin{itemize}
    \item \(\mathcal{P}\) is the partition of \(X_t\) induced by the witness sets
    \(\mathcal{W}\):
    \[
        \mathcal{P}
        =
        \{W\cap X_t\mid W\in\mathcal{W}
        \text{ and }W\cap X_t\neq\emptyset\}.
    \]

    \item \(\mathcal{C}\) is a partition of \(X_t\) refining
    \(\mathcal{P}\). For every \(P=W\cap X_t\in\mathcal{P}\), two vertices
    of \(P\) belong to the same part of \(\mathcal{C}\) if and only if they
    belong to the same connected component of \(G_t[W\cap V_t]\).

    \item \(\tau:\mathcal{P}\to
    (X_t\cap V(H))\cup\{\uparrow,\downarrow\}\)
    is a function that maps each part \(P\in\mathcal{P}\) to the status
    of its unique representative vertex \(h\in V(H)\):
    \[
    \tau(P)=
    \begin{cases}
        h & \text{if }h\in X_t
        \quad\text{(Representative in Bag)},\\
        \downarrow & \text{if }h\in V(H_t)\setminus X_t
        \quad\text{(Representative Forgotten)},\\
        \uparrow & \text{if }h\in V(H)\setminus V(H_t)
        \quad\text{(Representative Above)}.
    \end{cases}
    \]
\end{itemize}
\end{definition}

Thus, \(\mathcal{P}\) records which vertices of the current bag belong to the
same witness set, while \(\mathcal{C}\) records which of these vertices are
already connected in the processed part of that witness set. In particular,
a single part of \(\mathcal{P}\) may contain several parts of
\(\mathcal{C}\), which may merge when additional edges are introduced.

The function \(\tau\) records the location of the representative of each active
witness set. If the representative belongs to the bag, we record the vertex
itself. Otherwise, \(\downarrow\) and \(\uparrow\) distinguish whether the
representative has already been processed or is yet to appear above the current
node.

\begin{definition}[Adjacency Sets]
\label{def:new-pseudo-adjacency}
Consider a witness impression
\((\mathcal{P},\mathcal{C},\tau)\) at node \(t\).
For any pair of distinct parts
\((P_i,P_j)\in\binom{\mathcal{P}}{2}\), let \(h_i,h_j\) be their respective
representative vertices. We define the following two sets:
\begin{itemize}
    \item \(\mathcal{A}_t\) is the set of pseudo-adjacent pairs
    \((P_i,P_j)\) such that there is an edge in \(G_t\) between
    \(W_i\cap V_t\) and \(W_j\cap V_t\), where
    \(W_i\cap X_t=P_i\) and \(W_j\cap X_t=P_j\).

    \item \(\mathcal{R}_t\) is the set of anti-adjacent pairs
    \((P_i,P_j)\) such that the representatives \(h_i,h_j\) are
    non-adjacent in \(H\), i.e.,
    \((h_i,h_j)\notin E(H)\).
\end{itemize}
\end{definition}

In other words, \(\mathcal{A}_t\) stores the pairs of parts for which the
adjacency between the corresponding witness sets has already been witnessed by
some edge in \(G_t\). It is possible that both endpoints of such an edge lie
completely in \(V_t\setminus X_t\). Consider, for example, the witness sets
\(W_1\) and \(W_3\) in Figure~\ref{fig:side-by-side}.

Similarly, \(\mathcal{R}_t\) stores the pairs of witness sets that should remain
non-adjacent. This information is immediate when the corresponding
representatives are present in the bag, but becomes important once one of the
representatives has been forgotten. For example, in
Figure~\ref{fig:side-by-side}, the pair corresponding to \(W_1\) and \(W_2\)
belongs to \(\mathcal{R}_t\), as their representative vertices are
non-adjacent in \(H\). When the representative of a part has status \(\uparrow\), the corresponding
entries of \(\mathcal{R}_t\) are guessed by the dynamic program; these guesses
are verified when the representative is introduced, while a pair with statuses
\(\uparrow\) and \(\downarrow\) is necessarily anti-adjacent.

We are now in a position to describe our dynamic programming algorithm.
The dynamic programming algorithm maintains a table \(\texttt{dp}[t;S]\), where
the signature \(S\) at node \(t\) is the tuple
\(S=(\mathcal{P},\mathcal{C},\tau,\mathcal{A}_t,\mathcal{R}_t)\).
The partition \(\mathcal{P}\) records the active witness sets,
\(\mathcal{C}\) records the connected components of their processed portions,
\(\tau\) records the positions of their representatives,
\(\mathcal{A}_t\) records adjacencies already witnessed in \(G_t\), and
\(\mathcal{R}_t\) records adjacencies forbidden by \(H\).

We call a partition of \(V_t\) a \emph{partial witness structure} at \(t\)
if every part contains at most one vertex of \(V(H)\), every connected
component of a part intersecting \(X_t\) also intersects \(X_t\), and every
part disjoint from \(X_t\) is connected and already contains its representative.
Moreover, whenever a part becomes disjoint from \(X_t\), all its required
adjacencies with the remaining active witness sets must already have been
realized.

We set \(\texttt{dp}[t;S]=\true\) if and only if there exists a partial witness
structure on \(G_t\) realizing the information in \(S\). We additionally require
that every connected component of an active witness set intersects \(X_t\).
If a witness set no longer intersects \(X_t\), then it is required to be
connected, to already contain its representative, and to have all its required
adjacencies with the remaining active witness sets realized.

\begin{definition}[Valid Signature]
\label{def:valid-signature}
Consider a partition \(\mathcal{P}\) of \(X_t\), a partition
\(\mathcal{C}\) of \(X_t\), a function
\(\tau:\mathcal{P}\to(X_t\cap V(H))\cup\{\uparrow,\downarrow\}\),
and two sets
\(\mathcal{A}_t,\mathcal{R}_t\subseteq\binom{\mathcal{P}}{2}\).
The tuple \(S=(\mathcal{P},\mathcal{C},\tau,\mathcal{A}_t,\mathcal{R}_t)\)
is a \emph{valid signature} if all of the following conditions hold:
\begin{enumerate}
    \item \textbf{Connectivity Refinement:}
    \(\mathcal{C}\) refines \(\mathcal{P}\).

    \item \textbf{Cardinality:}
    No part \(P_i\in\mathcal{P}\) contains more than one vertex of \(V(H)\).
    Moreover, if \(P_i\cap V(H)=\{h\}\), then \(\tau(P_i)=h\).

    \item \textbf{Unique Mapping:}
    There do not exist two distinct parts \(P_i,P_j\in\mathcal{P}\) such that
    \(\tau(P_i)=\tau(P_j)\notin\{\uparrow,\downarrow\}\).

    \item \textbf{Representative Adjacency:}
    If \(\tau(P_i),\tau(P_j)\in V(H)\), then
    \((P_i,P_j)\in\mathcal{R}_t\) if and only if
    \((\tau(P_i),\tau(P_j))\notin E(H)\).
    Moreover, if
    \(\{\tau(P_i),\tau(P_j)\}=\{\uparrow,\downarrow\}\),
    then \((P_i,P_j)\in\mathcal{R}_t\).

    \item \textbf{Adjacency Condition:}
    \(\mathcal{A}_t\cap\mathcal{R}_t=\emptyset\).
\end{enumerate}
\end{definition}

An \emph{invalid signature} is a signature that does not satisfy the conditions
of Definition~\ref{def:valid-signature}.

\begin{lemma}
\label{lem:invalid-false}
For an invalid signature \(S\), \(\texttt{dp}[t;S]=\false\).
\end{lemma}

\begin{proof}
Suppose that \(S\) is induced by a partial witness structure.

If \(\mathcal{C}\) does not refine \(\mathcal{P}\), then some connected
component recorded by \(\mathcal{C}\) contains vertices belonging to different
witness sets, contradicting its definition.

If a part \(P_i\) contains two vertices of \(V(H)\), then the corresponding
witness set contains two representatives, which is impossible. Similarly, if
\(P_i\) contains \(h\in V(H)\), then \(h\) is necessarily the representative
of this witness set and hence \(\tau(P_i)=h\).

If distinct parts \(P_i,P_j\) map to the same representative
\(h\in V(H)\), then both must belong to the same witness set, contradicting
the fact that they are distinct parts of \(\mathcal{P}\).

If both representatives are present in the bag, membership in
\(\mathcal{R}_t\) is determined exactly by their adjacency in \(H\).
Furthermore, if one representative has already been forgotten and the other
has not yet appeared, then the two representatives cannot be adjacent in
\(H\): otherwise some bag of the tree decomposition of \(G\cup H\) would
contain both of them. Hence such a pair must belong to \(\mathcal{R}_t\).

Finally, if a pair belongs to both \(\mathcal{A}_t\) and
\(\mathcal{R}_t\), then the corresponding witness sets are adjacent in \(G_t\)
while their representatives are non-adjacent in \(H\), contradicting the
definition of a witness structure. Therefore an invalid signature cannot be
realized.
\end{proof}

In the rest of the section, whenever we mention a signature, we assume it to be
valid according to Definition~\ref{def:valid-signature}. We now describe how the
dynamic programming table is updated at each type of node.

\paragraph{Leaf Node}
The following lemma is immediate.

\begin{lemma}
\label{lem:dp-leaf-node-new}
Let \(t\) be a leaf node. The DP entry is
\(\texttt{dp}[t;S]=\true\) if and only if
\(S=(\emptyset,\emptyset,\emptyset,\emptyset,\emptyset)\).
\end{lemma}

\begin{proof}
Since \(X_t=\emptyset\), both partitions \(\mathcal{P}\) and
\(\mathcal{C}\) are empty. Consequently, \(\tau\),
\(\mathcal{A}_t\), and \(\mathcal{R}_t\) are empty as well.
\end{proof}

\paragraph{Introduce Vertex Node}

Let \(t\) be an \emph{introduce vertex} node with child \(t'\), such that
\(X_t=X_{t'}\cup\{x\}\).
As introduce-edge nodes are treated separately, none of the edges incident
with \(x\) has been introduced at this point. Consequently, \(x\) initially
forms a new connected component of whichever witness set it belongs to.

Let \(S=(\mathcal{P},\mathcal{C},\tau,\mathcal{A}_t,\mathcal{R}_t)\)
be a signature at \(t\), and let
\(S'=(\mathcal{P}',\mathcal{C}',\tau',
\mathcal{A}_{t'},\mathcal{R}_{t'})\)
be a signature at \(t'\).

\begin{definition}[Introduce-Vertex Compatibility]
\label{def:intro-vertex-comp}
The signatures \(S\) and \(S'\) are \emph{introduce-vertex compatible}
if there exists a mapping
\(\phi:\mathcal{P}'\to\mathcal{P}\) satisfying one of the following cases:

\begin{enumerate}[label=(\roman*)]
    \item \(x\) joins an existing witness set. There exists a unique
    \(P'_i\in\mathcal{P}'\) such that
    \(\phi(P'_i)=P'_i\cup\{x\}\),
    and for every other \(P'_j\in\mathcal{P}'\),
    \(\phi(P'_j)=P'_j\).
    \begin{itemize}
        \item[(C1)] \(\mathcal{C}=\mathcal{C}'\cup\{\{x\}\}\).

        \item[(C2)] If \(x\in V(H)\), then
        \(\tau'(P'_i)=\uparrow\) and \(\tau(\phi(P'_i))=x\).
        For all other parts the value of \(\tau\) remains unchanged.
        If \(x\notin V(H)\), then
        \(\tau(\phi(P'))=\tau'(P')\) for every
        \(P'\in\mathcal{P}'\).

        \item[(C3)]
        \(\mathcal{A}_t=
        \{(\phi(P'_a),\phi(P'_b))
        \mid (P'_a,P'_b)\in\mathcal{A}_{t'}\}\), and
        \(\mathcal{R}_t=
        \{(\phi(P'_a),\phi(P'_b))
        \mid (P'_a,P'_b)\in\mathcal{R}_{t'}\}\).
        If \(x\in V(H)\) and \(\tau'(P'_j)=\downarrow\), then
        \((\phi(P'_i),\phi(P'_j))\in\mathcal{R}_t\).
    \end{itemize}

    \item \(x\) starts a new witness set. Let \(P_x=\{x\}\). Then
    \(\mathcal{P}=\mathcal{P}'\cup\{P_x\}\) and
    \(\mathcal{C}=\mathcal{C}'\cup\{P_x\}\).
    \begin{itemize}
        \item[(C1)] If \(x\in V(H)\), then \(\tau(P_x)=x\).
        Otherwise, \(\tau(P_x)=\uparrow\).
        For every \(P\in\mathcal{P}'\), \(\tau(P)=\tau'(P)\).

        \item[(C2)] \(\mathcal{A}_t=\mathcal{A}_{t'}\).

        \item[(C3)] The restriction of \(\mathcal{R}_t\) to
        \(\binom{\mathcal{P}'}{2}\) equals \(\mathcal{R}_{t'}\).
        The pairs containing \(P_x\) are chosen so that the resulting
        signature \(S\) is valid.
        If \(x\in V(H)\) and \(\tau(P_j)=\downarrow\), then
        \((P_x,P_j)\in\mathcal{R}_t\).
    \end{itemize}
\end{enumerate}
\end{definition}

\begin{lemma}
\label{lem:dp-introduce-node}
Let \(t\) be an introduce vertex node with child \(t'\).
For a valid signature \(S\), \(\texttt{dp}[t;S]=\true\) if and only if
there exists a signature \(S'\) such that
\(\texttt{dp}[t';S']=\true\) and \(S,S'\) are introduce-vertex compatible.
\end{lemma}

\begin{proof}
Suppose first that \(\texttt{dp}[t;S]=\true\).
Removing \(x\) from its partial witness set gives the corresponding partial
witness structure at \(t'\). Since no edge incident with \(x\) has been
introduced, \(x\) forms a singleton connected component and no new pair is
added to \(\mathcal{A}_t\). Depending on whether the remaining part of its
witness set intersects \(X_{t'}\), we obtain Case~(i) or Case~(ii).

If \(x\in V(H)\), then \(x\) becomes the representative of its witness set.
If another representative has already been forgotten, it cannot be adjacent to
\(x\) in \(H\), since the two vertices never occur together in a bag of the
tree decomposition of \(G\cup H\). This explains the corresponding update of
\(\mathcal{R}_t\).

Conversely, given a partial witness structure realizing \(S'\), add \(x\) to
the witness set specified by the compatibility conditions, or let it start a
new witness set. As no incident edge of \(x\) has been introduced, the singleton
\(\{x\}\) is a new connected component intersecting the bag. All other
information is updated exactly as prescribed above.
\end{proof}

\paragraph{Introduce Edge Node}

Let \(t\) be an \emph{introduce edge} node with child \(t'\), and let
\(e=(u,v)\in E(G)\) be the edge introduced at \(t\).
If \(u\) and \(v\) belong to the same witness set, the new edge may merge two
connected components of its processed portion. If they belong to distinct
witness sets, the edge witnesses an adjacency between these two witness sets.

Let \(P_u,P_v\in\mathcal{P}\) denote the parts containing \(u,v\), respectively.

\begin{definition}[Introduce-Edge Compatibility]
\label{def:intro-edge-comp}
The signatures \(S\) and \(S'\) are \emph{introduce-edge compatible} if
\(\mathcal{P}=\mathcal{P}'\), \(\tau=\tau'\), and
\(\mathcal{R}_t=\mathcal{R}_{t'}\), and one of the following holds:
\begin{enumerate}[label=(\roman*)]
    \item If \(P_u=P_v\), let \(C_u,C_v\in\mathcal{C}'\) be the parts
    containing \(u,v\), respectively. If \(C_u=C_v\), then
    \(\mathcal{C}=\mathcal{C}'\); otherwise
    \(\mathcal{C}=
    (\mathcal{C}'\setminus\{C_u,C_v\})\cup\{C_u\cup C_v\}\).
    In either case, \(\mathcal{A}_t=\mathcal{A}_{t'}\).

    \item If \(P_u\neq P_v\), then
    \(\mathcal{C}=\mathcal{C}'\) and
    \(\mathcal{A}_t=
    \mathcal{A}_{t'}\cup\{(P_u,P_v)\}\).
\end{enumerate}
\end{definition}

\begin{lemma}
\label{lem:dp-introduce-edge}
Let \(t\) be an introduce edge node with child \(t'\), and let \(e=(u,v)\)
be the introduced edge. Then \(\texttt{dp}[t;S]=\true\) if and only if
there exists a signature \(S'\) such that
\(\texttt{dp}[t';S']=\true\) and \(S,S'\) are introduce-edge compatible.
\end{lemma}

\begin{proof}
If \(u,v\) belong to the same witness set, introducing the edge joins the
connected components containing \(u\) and \(v\), which is precisely the update
of \(\mathcal{C}\). If they belong to different witness sets, their connectivity
information is unchanged and the pair \((P_u,P_v)\) is added to
\(\mathcal{A}_t\). Since \(S\) is valid, this pair cannot belong to
\(\mathcal{R}_t\).
\end{proof}

\paragraph{Forget Vertex Node}

Let \(t\) be a \emph{forget vertex} node with child \(t'\), such that
\(X_t=X_{t'}\setminus\{x\}\).
The important distinction is whether the witness set containing \(x\) remains
represented in the bag. If it does, then the connected component containing
\(x\) must also continue to intersect the bag. If it does not, then the
corresponding witness set becomes completely processed, and its connectivity
and adjacency requirements must already be satisfied.

Let \(P_x\in\mathcal{P}'\) be the part containing \(x\), and let
\(C_x\in\mathcal{C}'\) be the part containing \(x\).

\begin{definition}[Forget-Vertex Compatibility]
\label{def:forget-node-comp}
The signatures \(S\) and \(S'\) are \emph{forget-vertex compatible} if one
of the following cases holds:

\begin{enumerate}[label=(\roman*)]
    \item \(x\) is not the only vertex in its witness part
    (\(P_x\neq\{x\}\)). Let \(P_x^*=P_x\setminus\{x\}\).
    We require \(C_x\neq\{x\}\).
    Let \(\phi(P_x)=P_x^*\) and \(\phi(P)=P\) for \(P\neq P_x\).
    \begin{itemize}
        \item[(C1)]
        \(\mathcal{P}=
        (\mathcal{P}'\setminus\{P_x\})\cup\{P_x^*\}\).

        \item[(C2)]
        \(\mathcal{C}=
        (\mathcal{C}'\setminus\{C_x\})\cup\{C_x\setminus\{x\}\}\).

        \item[(C3)] For all parts other than \(P_x^*\), the value of
        \(\tau\) remains unchanged. Moreover,
        \[
        \tau(P_x^*)=
        \begin{cases}
            \downarrow & \text{if }\tau'(P_x)=x,\\
            \tau'(P_x) & \text{otherwise}.
        \end{cases}
        \]

        \item[(C4)]
        \(\mathcal{A}_t=
        \{(\phi(P_a),\phi(P_b))
        \mid(P_a,P_b)\in\mathcal{A}_{t'}\}\), and similarly
        \(\mathcal{R}_t=
        \{(\phi(P_a),\phi(P_b))
        \mid(P_a,P_b)\in\mathcal{R}_{t'}\}\).
    \end{itemize}

    \item \(x\) is the last vertex in its witness part
    (\(P_x=\{x\}\)).
    \begin{itemize}
        \item[(C1)] \(\tau'(P_x)\neq\uparrow\).

        \item[(C2)] If \(\tau'(P_x)=x\) and
        \(\tau'(P_j)=\uparrow\), then
        \((P_x,P_j)\in\mathcal{R}_{t'}\)
        for every \(P_j\neq P_x\).

        \item[(C3)] For every
        \(P_j\in\mathcal{P}'\setminus\{P_x\}\),
        if \((P_x,P_j)\notin\mathcal{R}_{t'}\), then
        \((P_x,P_j)\in\mathcal{A}_{t'}\).

        \item[(C4)] The parent signature is obtained by deleting
        \(P_x\) from \(\mathcal{P}'\), deleting \(\{x\}\) from
        \(\mathcal{C}'\), and restricting
        \(\tau'\), \(\mathcal{A}_{t'}\), and \(\mathcal{R}_{t'}\)
        to the remaining parts.
    \end{itemize}
\end{enumerate}
\end{definition}

\begin{lemma}
\label{lem:dp-forget-node}
Let \(t\) be a forget vertex node with child \(t'\). For a valid signature \(S\),
\(\texttt{dp}[t;S]=\true\) if and only if there exists a signature \(S'\)
such that \(\texttt{dp}[t';S']=\true\) and
\(S,S'\) are forget-vertex compatible.
\end{lemma}

\begin{proof}
Suppose first that \(P_x\neq\{x\}\). The witness set containing \(x\)
remains active. Therefore, the connected component containing \(x\) must still
intersect \(X_t\). This is exactly the condition \(C_x\neq\{x\}\).
The remaining conditions simply restrict the two partitions to the new bag and
update the representative status if \(x\) is itself the representative.

Suppose now that \(P_x=\{x\}\). The corresponding witness set ceases to
intersect the bag. Since every connected component of an active witness set is
required to intersect the bag, this witness set is connected. Its representative
must already have appeared, giving (C1). Furthermore, no edge involving this
witness set can be introduced above \(t\). Hence every adjacency required by
\(H\) with a remaining active witness set must already have been witnessed in
\(G_{t'}\), giving (C3). Pairs belonging to \(\mathcal{R}_{t'}\) remain
non-adjacent because the signature is valid. Condition (C2) handles the case
where a representative still lies above the current node.

The reverse direction follows by applying the same observations to a partial
witness structure realizing \(S'\).
\end{proof}

\paragraph{Join Node}

Let \(t\) be a \emph{join node} with children \(t_L,t_R\), such that
\(X_t=X_{t_L}=X_{t_R}\).
The two child subgraphs meet precisely on the common bag. Hence they must agree
on the witness-set partition of the bag, while their connectivity and witnessed
adjacency information are combined at \(t\).

For two partitions \(\mathcal{C}_L,\mathcal{C}_R\) of \(X_t\), let
\(\mathcal{C}_L\vee\mathcal{C}_R\) denote the partition corresponding to the
transitive closure of the union of their equivalence relations.

\begin{definition}[Join-Node Compatibility]
\label{def:join-node-comp}
Valid signatures
\(S=(\mathcal{P},\mathcal{C},\tau,\mathcal{A}_t,\mathcal{R}_t)\),
\(S_L=(\mathcal{P}_L,\mathcal{C}_L,\tau_L,\mathcal{A}_L,\mathcal{R}_L)\),
and
\(S_R=(\mathcal{P}_R,\mathcal{C}_R,\tau_R,\mathcal{A}_R,\mathcal{R}_R)\)
are \emph{join-compatible} if:
\begin{itemize}
    \item[(C1)] \(\mathcal{P}=\mathcal{P}_L=\mathcal{P}_R\).

    \item[(C2)] \(\mathcal{C}=\mathcal{C}_L\vee\mathcal{C}_R\).

    \item[(C3)] For every \(P\in\mathcal{P}\), one of the following holds:
    \[
    \begin{array}{c|c|c}
        \tau_L(P) & \tau_R(P) & \tau(P)\\ \hline
        h & h & h\\
        \uparrow & \uparrow & \uparrow\\
        \downarrow & \uparrow & \downarrow\\
        \uparrow & \downarrow & \downarrow
    \end{array}
    \]
    where \(h\in P\cap V(H)\).

    \item[(C4)] \(\mathcal{A}_t=\mathcal{A}_L\cup\mathcal{A}_R\).

    \item[(C5)] \(\mathcal{R}_t=\mathcal{R}_L=\mathcal{R}_R\).
\end{itemize}
\end{definition}

Notice in particular that
\(\tau_L(P)=\tau_R(P)=\downarrow\) is not possible. Indeed, in that case
the representative would occur strictly inside both child subtrees while being
absent from their common bag, contradicting the connectedness property of a
tree decomposition.

\begin{lemma}
\label{lem:dp-join-node}
Let \(t\) be a join node with children \(t_L,t_R\). For a valid signature \(S\),
\(\texttt{dp}[t;S]=\true\) if and only if there exist signatures \(S_L,S_R\)
such that \(\texttt{dp}[t_L;S_L]=\texttt{dp}[t_R;S_R]=\true\) and
\(S,S_L,S_R\) are join-compatible.
\end{lemma}

\begin{proof}
The bag partitions induced by the two child partial witness structures are
identical. Two vertices in the same witness set are connected in \(G_t\) if
their connectivity follows from the left child, the right child, or a sequence
alternating between the two through vertices of the common bag. This gives
\(\mathcal{C}=\mathcal{C}_L\vee\mathcal{C}_R\).

A representative contained in \(X_t\) occurs in both children. A representative
strictly contained in one child subtree is forgotten in that child and above in
the other. These are exactly the cases in (C3).

Finally, an adjacency between two active witness sets may be witnessed in either
child, so \(\mathcal{A}_t=\mathcal{A}_L\cup\mathcal{A}_R\).
The anti-adjacency information concerns the same eventual representatives and
therefore has to agree in the two children. This proves the claimed equivalence.
\end{proof}

\begin{lemma}
\label{lem:dp-root}
Let \(r\) be the root of \(\mathcal{T}\). Then \(H\) is a labeled contraction
of \(G\) if and only if
\(\texttt{dp}[r;(\emptyset,\emptyset,\emptyset,\emptyset,\emptyset)]=\true\).
\end{lemma}

\begin{proof}
At the root, \(X_r=\emptyset\) and \(G_r=G\). Hence every witness set has
already ceased to be active. By the forget transition, each such witness set is
connected, contains exactly one representative, and has all required adjacencies
to the other witness sets realized. Moreover, validity of the signatures ensures
that no adjacency corresponding to a non-edge of \(H\) is created. Thus the
resulting partition is an \(H\)-witness structure of \(G\).

Conversely, let \(\mathcal{W}\) be an \(H\)-witness structure of \(G\).
For every node \(t\), the intersections of its witness sets with \(V_t\)
induce the corresponding signature. Since each witness set is connected, every
connected component of its processed portion that has not yet been completed
must intersect the current bag. Hence these signatures satisfy the transitions
described above, and the empty signature at the root is accepted.
\end{proof}

We are finally in a position to present the proof of the algorithm mentioned in
Theorem~\ref{thm:treewith}.

\begin{proof}[Algorithm in Theorem~\ref{thm:treewith}]
The algorithm employs dynamic programming over a nice tree decomposition of
\(G\cup H\), maintaining states indexed by signatures
\(S=(\mathcal{P},\mathcal{C},\tau,\mathcal{A}_t,\mathcal{R}_t)\).
For a bag \(X_t\) of size \(\mathcal{O}(\tw)\), the number of partitions
\(\mathcal{P}\) is bounded by \(2^{\mathcal{O}(\tw\log\tw)}\).
The number of possible refinements \(\mathcal{C}\) is also bounded by
\(2^{\mathcal{O}(\tw\log\tw)}\).
The status function \(\tau\), which maps each part to one of
\(\mathcal{O}(\tw)\) possible representatives or to one of
\(\uparrow,\downarrow\), has \(2^{\mathcal{O}(\tw\log\tw)}\)
possibilities.

Since there are \(\mathcal{O}(\tw^2)\) pairs of parts, the total number of
possibilities for \(\mathcal{A}_t\) and \(\mathcal{R}_t\) is
\(2^{\mathcal{O}(\tw^2)}\). Therefore, the number of signatures at each node is
\(2^{\mathcal{O}(\tw^2)}\).

Each transition at Leaf, Introduce Vertex, Introduce Edge, and Forget Vertex
nodes can be checked in time polynomial in \(\tw\) for a pair of signatures.
At Join nodes, pairs of child signatures are considered; this still requires
\(2^{\mathcal{O}(\tw^2)}\) time. Since the nice tree decomposition contains
a polynomial number of nodes, the overall running time is
\(2^{\mathcal{O}(\tw^2)}\cdot |V(G)|^{\mathcal{O}(1)}\).
Correctness follows from
Lemmas~\ref{lem:dp-leaf-node-new},
\ref{lem:dp-introduce-node},
\ref{lem:dp-introduce-edge},
\ref{lem:dp-forget-node},
\ref{lem:dp-join-node}, and~\ref{lem:dp-root}.
\end{proof}

\subsection{Conditional Lower Bound}

In this section, we demonstrate that unless the \ETH\ 
fails, \textsc{Labeled Contraction} does not admit an algorithm running in time 
$2^{o(\tw^2)} \cdot |V(G)|^{\mathcal{O}(1)}$, where $\tw$ is the treewidth 
of the input graph $G$. To obtain the lower bound, we reduce from a special case 
of \textsc{Vertex Cover} on sub-cubic graphs, defined in~\cite{agrawal2019split} 
as follows:

\defproblem{\textsc{Sub-Cubic Partitioned Vertex Cover (Sub-Cubic PVC)}}{A sub-cubic 
graph $G$; an integer $t$; 
for $i \in [t]$, an integer $k_i \ge 0$; 
a partition $\mathcal{P} = \{C_1, \ldots, C_t\}$ of $V(G)$ such that 
$t \in \mathcal{O}(\sqrt{|V(G)|})$ and for all $i \in [t]$, 
$C_i$ is an independent set and $|C_i| \in \mathcal{O}(\sqrt{|V(G)|})$. 
Furthermore, for $i,j \in [t]$, $i \ne j$, 
$|E(G[C_i \cup C_j]) \cap E(G)| = 1$.}{Does $G$ have a vertex cover 
$X$ such that for all $i \in [t]$, $|X \cap C_i| \le k_i$?}

\begin{proposition}[Theorem~3.9 in~\cite{agrawal2019split}]
\label{prop:sub-cubic-pvc-eth}
\textsc{Sub-Cubic PVC} does not admit an algorithm running 
in time $2^{o(n)}$, unless the \ETH\ fails.
\end{proposition}

We now prove the lower bound in Theorem~\ref{thm:treewith}.

Let \((G, \mathcal{P} = \{C_1, C_2, \ldots, C_t\}, k_1, \ldots, k_t)\) 
be an instance of the \textsc{Sub-Cubic PVC} problem, where \(n = |V(G)|\). 
We construct an instance \((G', H')\) of the \textsc{Labeled Contraction} 
problem as follows.

Before giving the formal construction, we briefly describe its main idea.
For every part $C_i$, the vertices $x_i$ and $y_i$ represent the two choices
for a vertex of $C_i$: being included in the vertex cover or not, respectively.
The set $S_i$ ensures that at most $k_i$ vertices of $C_i$ are contracted to
$x_i$. Finally, the unique edge between every pair $C_i,C_j$ is encoded so
that the required edge $(x_i,x_j)$ in $H'$ can be created only if at least
one endpoint of this edge belongs to the vertex cover.

\noindent\emph{Construction of \(G'\):}
We begin by setting \(V(G') \supseteq V(G)\), i.e., all vertices of 
\(G\) are copied into \(G'\) without copying any edges. 
For each part \(C_i \in \mathcal{P}\), we introduce a set \(S_i\) consisting of 
\(\binom{|C_i|}{k_i + 1}\) new vertices in \(G'\), where each vertex 
\(s \in S_i\) corresponds uniquely to a subset \(C_i' \subseteq C_i\) of size 
exactly \(k_i+1\). For each such vertex \(s\in S_i\), we add edges between 
\(s\) and every vertex in \(C_i'\).

Additionally, for every \(C_i \in \mathcal{P}\), we introduce three auxiliary 
vertices \(\{x_i,y_i,z_i\}\) in \(G'\). Each of these vertices is connected 
to every vertex in \(C_i\). Moreover, we add the edges 
\((z_i,x_i),(z_i,y_i)\), and \((x_i,y_i)\). Each vertex \(s\in S_i\) is 
also connected to \(x_i\).

We now encode the edges of \(G\). Consider distinct \(i,j\in[t]\), and let
\((u,v)\) be the unique edge of \(G\) with \(u\in C_i\) and \(v\in C_j\).
We add the edges \((u,x_j)\) and \((v,x_i)\) to \(G'\).

Finally, the set \(\{y_1,y_2,\ldots,y_t\}\) is made into a clique, and for 
every distinct \(i,j\in[t]\), we add the edge \((x_i,y_j)\) in \(G'\).
Together with the already added edges \((x_i,y_i)\), this makes every
$x_i$ adjacent to every $y_j$.

\noindent\emph{Construction of \(H'\):}
We have
\[
V(H') =
\left(\bigcup_{i=1}^{t}S_i\right)
\cup
\left(\bigcup_{i=1}^{t}\{x_i,y_i,z_i\}\right).
\]
Thus, \(H'\) contains all newly introduced vertices but none of the vertices 
of the original graph \(G\).

The edge set of \(H'\) is defined as follows. First, we form a clique on
\(\{x_1,\ldots,x_t\}\cup\{y_1,\ldots,y_t\}\). For every \(i\in[t]\),
we add the edges \((x_i,z_i)\) and \((y_i,z_i)\). Finally, for every
\(s\in S_i\), we add both edges \((s,x_i)\) and \((s,y_i)\).
There are no other edges in \(H'\).

This concludes the construction.
We refer the reader to Figure~\ref{fig:treewidth-lower-bound} 
for an illustration of the reduction.
We now prove the correctness of the reduction.

\begin{figure}[t]
\centering

\begin{subfigure}[t]{0.45\textwidth}
\centering
\begin{tikzpicture}[
scale = 0.85,
dot/.style={circle, fill, inner sep=1pt, scale = 0.8},
sq/.style={rectangle, draw, minimum size=4mm},
sball/.style={ellipse, draw, minimum height=1.2cm, minimum width=3.5cm},
edge/.style={thick, purple},
vc/.style={orange, thick},
subs/.style={violet, thick},
new/.style={green, thick},
side/.style={gray, thick},
every node/.style={font=\small},
]


\node[ellipse, draw, minimum height=1cm, minimum width=3cm, scale = 0.8] (C1) at (0,0.5) {};
\node at ($(C1)+(-2,0)$) {$C_1$};

\node[ellipse, draw, minimum height=1.2cm, minimum width=3cm, scale = 0.8] (S1) at (0,2.5) {};
\node at ($(S1)+(-2,0)$) {$S_1$};

\node[dot, label=above:$s_1$] (s1) at ($(S1)+(-0.9,-0.2)$) {};
\node[dot, label=above:$s_2$] (s2) at ($(S1)+(-0.3,-0.2)$) {};
\node[dot, label=above:$\cdots$] (s3) at ($(S1)+(0.3,-0.2)$) {};
\node[dot, label=above:$s_p$] (s4) at ($(S1)+(0.9,-0.2)$) {};

\node[dot, label={[xshift=1pt, yshift=2pt, scale = 0.8]right:$v_1$}] (v1) at ($(C1)+(-1,-0.1)$) {};
\node[dot, label={[xshift=1pt, yshift=2pt, scale = 0.8]right:$\cdots$}] (v2) at ($(C1)+(-0.2,-0.1)$) {};
\node[dot, label={[xshift=1pt, yshift=2pt, scale = 0.8]right:$v_{n_1}$}] (v3) at ($(C1)+(0.6,-0.1)$) {};

\draw[subs] (s1) -- (v1);
\draw[subs] (s1) -- (v2);
\draw[subs] (s2) -- (v1);
\draw[subs] (s2) -- (v3);

\node[ellipse, draw, minimum height=7.5mm, minimum width=7.5mm, scale = 0.8] (x1) at (1,-1.5) {$x_1$};
\node[ellipse, draw, minimum height=7.5mm, minimum width=7.5mm, scale = 0.8] (y1) at (0,-1.5) {$y_1$};
\node[ellipse, draw, minimum height=7.5mm, minimum width=7.5mm, scale = 0.8] (z1) at (-1,-1.5) {$z_1$};
\draw[edge] (v1) -- (x1);
\draw[edge] (v2) -- (x1);
\draw[edge] (v3) -- (x1);
\draw[edge] (v1) -- (y1);
\draw[edge] (v2) -- (y1);
\draw[edge] (v3) -- (y1);
\draw[edge] (v1) -- (z1);
\draw[edge] (v2) -- (z1);
\draw[edge] (v3) -- (z1);
\draw[edge] (x1) -- (y1);

\node[ellipse, draw, minimum height=1cm, minimum width=3cm, scale = 0.8] (C2) at (4.5,0.5) {};
\node at ($(C2)+(2,0)$) {$C_2$};

\node[ellipse, draw, minimum height=1.2cm, minimum width=3cm, scale = 0.8] (S2) at (4.5,2.5) {};
\node at ($(S2)+(2,0)$) {$S_2$};

\node[dot, label=above:$s_1'$] (s1p) at ($(S2)+(-0.9,-0.2)$) {};
\node[dot, label=above:$s_2'$] (s2p) at ($(S2)+(-0.3,-0.2)$) {};
\node[dot, label=above:$\cdots$] (s3p) at ($(S2)+(0.3,-0.2)$) {};
\node[dot, label=above:$s_q'$] (skp) at ($(S2)+(0.9,-0.2)$) {};

\node[dot, label={[xshift=1pt, yshift=2pt, scale = 0.8]right:$v_1$}] (v4) at ($(C2)+(-1,-0.1)$) {};
\node[dot, label={[xshift=1pt, yshift=2pt, scale = 0.8]right:$\cdots$}] (v5) at ($(C2)+(-0.2,-0.1)$) {};
\node[dot, label={[xshift=1pt, yshift=2pt, scale = 0.8]right:$v_{n_2}$}] (v6) at ($(C2)+(0.6,-0.1)$) {};

\draw[subs] (s1p) -- (v4);
\draw[subs] (s1p) -- (v5);
\draw[subs] (s2p) -- (v5);
\draw[subs] (s2p) -- (v6);

\node[ellipse, draw, minimum height=7.5mm, minimum width=7.5mm, scale = 0.8] (x2) at (3.5,-1.5) {$x_2$};
\node[ellipse, draw, minimum height=7.5mm, minimum width=7.5mm, scale = 0.8] (y2) at (4.5,-1.5) {$y_2$};
\node[ellipse, draw, minimum height=7.5mm, minimum width=7.5mm, scale = 0.8] (z2) at (5.5,-1.5) {$z_2$};
\draw[edge] (v4) -- (x2);
\draw[edge] (v5) -- (x2);
\draw[edge] (v6) -- (x2);
\draw[edge] (v4) -- (y2);
\draw[edge] (v5) -- (y2);
\draw[edge] (v6) -- (y2);
\draw[edge] (v4) -- (z2);
\draw[edge] (v5) -- (z2);
\draw[edge] (v6) -- (z2);
\draw[edge] (x2) -- (y2);
\draw[edge, bend right = 40] (y1) to (y2);
\draw[edge] (y1) to (z1);
\draw[edge, bend right = 40] (z1) to (x1);
\draw[edge] (y2) to (z2);
\draw[edge, bend left = 40] (z2) to (x2);
\draw[vc] (v3) -- (x2);
\draw[vc] (v4) -- (x1);

\draw[side, bend left = 60] (S1) to (x1);
\draw[side, bend right = 60] (S2) to (x2);
\draw[edge, bend right = 30] (x1) to (y2);
\draw[edge, bend left = 30] (x2) to (y1);

\end{tikzpicture}
\caption{Graph $G'$}
\label{fig:G}
\end{subfigure}
\hfill
\begin{subfigure}[t]{0.45\textwidth}
\centering
\begin{tikzpicture}[
scale = 0.85,
dot/.style={circle, fill, inner sep=1pt, scale = 0.8},
edge/.style={thick, purple},
vc/.style={orange, thick},
subs/.style={violet, thick},
new/.style={olive, thick},
side/.style={darkgray, thick},
every node/.style={font=\small},
]


\node[ellipse, draw, minimum height=1.2cm, minimum width=3cm, scale = 0.8] (S1) at (0,2.5) {};
\node at ($(S1)+(0,1)$) {$S_1$};

\node[dot, label=above:$s_1$] (s1) at ($(S1)+(-0.9,-0.2)$) {};
\node[dot, label=above:$s_2$] (s2) at ($(S1)+(-0.3,-0.2)$) {};
\node[dot, label=above:$\cdots$] (s3) at ($(S1)+(0.3,-0.2)$) {};
\node[dot, label=above:$s_p$] (s4) at ($(S1)+(0.9,-0.2)$) {};

\node[ellipse, draw, minimum height=7.5mm, minimum width=7.5mm, scale = 0.8] (x1) at (1,0) {$x_1$};
\node[ellipse, draw, minimum height=7.5mm, minimum width=7.5mm, scale = 0.8] (y1) at (0,0) {$y_1$};
\node[ellipse, draw, minimum height=7.5mm, minimum width=7.5mm, scale = 0.8] (z1) at (-1,0) {$z_1$};
\draw[edge] (x1) -- (y1);

\draw[side] (s1) -- (x1);
\draw[side] (s2) -- (x1);
\draw[side] (s4) -- (x1);
\draw[new] (s1) -- (y1);
\draw[new] (s2) -- (y1);

\node[ellipse, draw, minimum height=1.2cm, minimum width=3cm, scale = 0.8] (S2) at (4,2.5) {};
\node at ($(S2)+(0,1)$) {$S_2$};

\node[dot, label=above:$s_1'$] (s1p) at ($(S2)+(-0.9,-0.2)$) {};
\node[dot, label=above:$s_2'$] (s2p) at ($(S2)+(-0.3,-0.2)$) {};
\node[dot, label=above:$\cdots$] (s3p) at ($(S2)+(0.3,-0.2)$) {};
\node[dot, label=above:$s_q'$] (skp) at ($(S2)+(0.9,-0.2)$) {};

\node[ellipse, draw, minimum height=7.5mm, minimum width=7.5mm, scale = 0.8] (x2) at (3,0) {$x_2$};
\node[ellipse, draw, minimum height=7.5mm, minimum width=7.5mm, scale = 0.8] (y2) at (4,0) {$y_2$};
\node[ellipse, draw, minimum height=7.5mm, minimum width=7.5mm, scale = 0.8] (z2) at (5,0) {$z_2$};
\draw[edge] (x2) -- (y2);

\draw[edge, bend right = 40] (y1) to (y2);

\draw[new, bend left = 40] (x1) to (z1);
\draw[side] (y1) to (z1);

\draw[new, bend right = 40] (x2) to (z2);
\draw[side] (y2) to (z2);

\draw[side] (s1p) -- (x2);
\draw[side] (s2p) -- (x2);
\draw[side] (skp) -- (x2);
\draw[new] (s1p) -- (y2);
\draw[new] (s2p) -- (y2);
\draw[new] (x1) -- (x2);

\draw[edge, bend right = 30] (x1) to (y2);
\draw[edge, bend left = 30] (x2) to (y1);

\end{tikzpicture}
\caption{Graph $H'$}
\label{fig:H'}
\end{subfigure}

\caption{Illustration of the graphs $G'$ and $H'$ in the reduction for the
lower bound of Theorem~\ref{thm:treewith}. Each $S_i$ is an independent set.
The $y$-vertices form a clique in $G'$, every $x_i$ is adjacent to every
$y_j$, and the set $\{x_1,\ldots,x_t,y_1,\ldots,y_t\}$ forms a clique in
$H'$.}
\label{fig:treewidth-lower-bound}
\end{figure}

\noindent $(\Rightarrow)$      
Suppose we are given a solution to the \textsc{Sub-Cubic PVC} instance, 
that is, a vertex cover \(X\subseteq V(G)\) such that for every 
\(i\in[t]\), it holds that \(|X\cap C_i|\le k_i\).
We describe a contraction sequence on \(G'\) that yields \(H'\).
For every \(i\in[t]\), we contract the edge \((x_i,a)\) for every
\(a\in X\cap C_i\), and contract the edge \((y_i,b)\) for every
\(b\in C_i\setminus X\).

By the definition of \textsc{Sub-Cubic PVC}, for each pair of distinct
\(i,j\in[t]\), there is exactly one edge in \(G\) between \(C_i\) and
\(C_j\), say \((u,v)\), where \(u\in C_i\) and \(v\in C_j\).
Since \(X\) is a vertex cover, \(u\in X\) or \(v\in X\).
If \(u\in X\), then contracting \(u\) into \(x_i\), together with the edge
\((u,x_j)\), creates the edge \((x_i,x_j)\). Similarly, if \(v\in X\),
then contracting \(v\) into \(x_j\), together with the edge \((v,x_i)\),
creates \((x_i,x_j)\). Hence the vertices \(x_1,\ldots,x_t\) form a clique
after the contractions. Since all edges between the \(x\)-vertices and
\(y\)-vertices, and all edges among the \(y\)-vertices, are already present
in \(G'\), the set
\(\{x_1,\ldots,x_t,y_1,\ldots,y_t\}\) induces the required clique.

Furthermore, for each \(i\in[t]\), since \(|X\cap C_i|\le k_i\), no subset
of \(C_i\) of size \(k_i+1\) is entirely contracted into \(x_i\).
Therefore, for every \(s\in S_i\), corresponding to a subset
\(C_i'\subseteq C_i\) of size \(k_i+1\), at least one vertex of \(C_i'\)
is contracted into \(y_i\). Since \(s\) is adjacent to every vertex of
\(C_i'\), this creates the edge \((s,y_i)\). The edge \((s,x_i)\) was
already present in \(G'\). Similarly, the edges \((x_i,z_i)\) and
\((y_i,z_i)\) are already present.

It remains to verify that no additional edge is created. Every contracted
vertex \(c\in C_i\) is contracted into either \(x_i\) or \(y_i\).
Apart from \(x_i,y_i,z_i\), its neighbors in \(G'\) are vertices of \(S_i\)
corresponding to subsets containing \(c\), and possibly vertices \(x_j\)
arising from edges of \(G\) incident with \(c\). Consequently, every edge
created by contracting \(c\) has the form
\((x_i,z_i)\), \((y_i,z_i)\), \((x_i,s)\), \((y_i,s)\),
\((x_i,x_j)\), or \((y_i,x_j)\), all of which are edges of \(H'\).
Thus, the resulting graph is precisely \(H'\).

\noindent $(\Leftarrow)$
Suppose that \(H'\) is a labeled contraction of \(G'\), and let
\(\mathcal{W}\) be a corresponding witness structure.
For \(r\in V(H')\), let \(W(r)\) denote the witness set whose representative
is \(r\).

We first record some consequences of the construction. The vertices of
\(V(G)\) form an independent set in \(G'\). Hence, if a vertex
\(c\in C_i\) belongs to a witness set \(W(r)\), then \(c\) must be adjacent
to the representative \(r\), as there is no edge between two vertices of
\(V(G)\).

Moreover, no vertex of \(C_i\) can belong to a witness set represented by a
vertex \(s\in S_i\): indeed, every vertex of \(C_i\) is adjacent to \(z_i\),
whereas \(s\) and \(z_i\) are non-adjacent in \(H'\). Similarly, a vertex
\(c\in C_i\) cannot belong to \(W(x_j)\) for \(j\ne i\). Such an adjacency
\(cx_j\) exists only when \(c\) is an endpoint of the unique edge between
\(C_i\) and \(C_j\); but then the edge \(cz_i\) would imply the forbidden
edge \(x_jz_i\) in \(H'\).

If \(k_i<|C_i|\), then every vertex \(c\in C_i\) belongs to at least one
subset of \(C_i\) of size \(k_i+1\), and hence is adjacent to some vertex
of \(S_i\). It follows that \(c\) cannot belong to \(W(z_i)\), since this
would create an edge between \(z_i\) and that vertex of \(S_i\), which is
not present in \(H'\). Therefore, in this case every vertex of \(C_i\)
belongs to either \(W(x_i)\) or \(W(y_i)\). If \(k_i\ge |C_i|\), the
condition \(|X\cap C_i|\le k_i\) will be automatic and no such conclusion
is needed for isolated vertices of \(G\).

We now define
\(X\subseteq V(G)\) by putting \(c\in X\cap C_i\) if and only if
\(c\in W(x_i)\).

We first prove that \(|X\cap C_i|\le k_i\). If \(k_i\ge |C_i|\), this is
immediate. Otherwise, consider a vertex \(s\in S_i\), corresponding to a
subset \(C_i'\subseteq C_i\) of size \(k_i+1\). The witness set \(W(s)\)
contains no vertex of \(V(G)\), since any such vertex would be adjacent to
\(z_i\), while \(s\) and \(z_i\) are non-adjacent in \(H'\). Hence
\(W(s)=\{s\}\). Since \((s,y_i)\in E(H')\), there must be an edge in \(G'\)
between \(s\) and \(W(y_i)\). The neighbors of \(s\) in \(V(G)\) are exactly
the vertices of \(C_i'\). Therefore, at least one vertex of \(C_i'\) belongs
to \(W(y_i)\). Thus no subset of \(C_i\) of size \(k_i+1\) is entirely
contained in \(X\), and consequently \(|X\cap C_i|\le k_i\).

It remains to show that \(X\) is a vertex cover of \(G\). Consider the unique
edge \((u,v)\) between \(C_i\) and \(C_j\), where \(u\in C_i\) and
\(v\in C_j\). Since \((x_i,x_j)\in E(H')\), the witness sets
\(W(x_i)\) and \(W(x_j)\) must be adjacent in \(G'\).
By the construction, \(W(x_i)\) contains no vertices outside
\(\{x_i\}\cup C_i\), and similarly \(W(x_j)\) contains no vertices outside
\(\{x_j\}\cup C_j\). There is no edge \((x_i,x_j)\) in \(G'\), and there
are no edges between vertices of \(C_i\) and \(C_j\) in \(G'\).
The only possible edges between these two witness sets are therefore
\((u,x_j)\) and \((v,x_i)\). Hence either \(u\in W(x_i)\) or
\(v\in W(x_j)\), implying that either \(u\in X\) or \(v\in X\).
Thus \(X\) intersects every edge of \(G\), and is a vertex cover satisfying
\(|X\cap C_i|\le k_i\) for every \(i\in[t]\).

This completes the proof of correctness of the reduction.

To upper bound the treewidth of \(G'\), define
\(W=\{x_i,y_i,z_i\mid i\in[t]\}\).
The graph \(G'-W\) consists of \(t\) vertex-disjoint graphs
\(G'[C_i\cup S_i]\). Each such graph is bipartite with bipartition
\((C_i,S_i)\), and \(C_i\) is a vertex cover of
\(G'[C_i\cup S_i]\). Hence
\(\tw(G'[C_i\cup S_i])\le |C_i|=\mathcal{O}(\sqrt n)\).
Since \(|W|=3t=\mathcal{O}(\sqrt n)\), adding \(W\) to every bag gives
\[
\tw(G')\le |W|+\max_{i\in[t]}\tw(G'[C_i\cup S_i])
=\mathcal{O}(\sqrt n).
\]

Now suppose, for the sake of contradiction, that there exists an algorithm 
\(\mathcal{A}\) that solves \textsc{Labeled Contraction} in time 
\(2^{o(\tw(G')^2)}\cdot |V(G')|^{\mathcal{O}(1)}\).
We use \(\mathcal{A}\) as a subroutine to construct an algorithm
\(\mathcal{B}\) for solving \textsc{Sub-Cubic PVC}.

On input
\((G,\mathcal{P}=\{C_1,\ldots,C_t\},k_1,\ldots,k_t)\),
where \(n=|V(G)|\), algorithm \(\mathcal{B}\) constructs the corresponding
instance \((G',H')\). Since
\(|S_i|\le 2^{|C_i|}=2^{\mathcal{O}(\sqrt n)}\) and
\(t=\mathcal{O}(\sqrt n)\), the construction takes
\(2^{\mathcal{O}(\sqrt n)}\) time and
\(|V(G')|=2^{\mathcal{O}(\sqrt n)}\).
It then invokes algorithm \(\mathcal{A}\) on \((G',H')\) and returns its result.

The correctness of \(\mathcal{B}\) follows from the correctness of the
reduction. Since \(\tw(G')=\mathcal{O}(\sqrt n)\), the running time of
\(\mathcal{A}\) on the constructed instance is
\(2^{o(n)}\cdot 2^{\mathcal{O}(\sqrt n)}=2^{o(n)}\).
Including the construction time, the total running time of
\(\mathcal{B}\) is \(2^{o(n)}\), contradicting the \ETH\ by
Proposition~\ref{prop:sub-cubic-pvc-eth}.
This establishes the claimed conditional lower bound in
Theorem~\ref{thm:treewith}.

\section{\NP-hardness When Maximum Degree is Bounded}
\label{sec:np-hard-max-degree-bound}

In this section, we prove Theorem~\ref{thm:bounded-max-degree} by establishing the \NP-hardness of the \textsc{Labeled Contractibility} problem. 
We restate it for the reader's conveience.

\NPhardBoundedMaxDegree*

We present a polynomial-time reduction from \textsc{Positive-Not-All-$(3,4)$-SAT} to \textsc{Labeled Contractibility}, even when both input graphs \(G\) and \(H\) have constant maximum degree. The source problem is formally defined as follows and is known to be \NP-complete~\cite{DBLP:journals/tcs/DarmannD20}.

\defproblem{\textsc{Positive-Not-All-$(3, 4)$-SAT}}{A $3$-CNF Boolean formula $\psi$ with a set of variables \( \{x_1, \dots, x_n\} \) and a set of clauses \( \{C_1, \dots, C_m\} \), where every literal is positive and each variable appears in at most 4 clauses.}{Does there exist a truth assignment such that every clause is satisfied and, in each clause, not all three literals are assigned \true?}

We may assume that every variable occurs in at least one clause, as variables that do not occur can simply be deleted. We further assume without loss of generality that \(2m+2\) is a power of two. Indeed, if this is not the case, we add clauses on fresh variables until the number of clauses is of the form \(2^{q-1}-1\). Each added clause uses three fresh variables and is therefore independently satisfiable with both a true and a false variable; moreover, every new variable occurs only once. Hence this padding preserves both satisfiability and the bound of four occurrences per variable.

Let \(2m+2=2^q\), and set \(p=2q+2\). We now construct two vertex-labeled graphs \(G\) and \(H\) from the formula \(\psi\).

\emph{Construction of \(H\):} The target graph \(H\) is the path on \(p\) vertices labeled \(v_1,v_2,\ldots,v_p\).

\emph{Construction of \(G\):} We build \(G\) using two symmetric binary-tree structures.

\begin{enumerate}
    \item \textbf{First Binary Tree:} Construct a complete binary tree \(T_L\) of depth \(q\), rooted at \(v_1\). Choose one distinguished leaf and label the vertices on the root-to-leaf path, in order, by \(v_1,v_2,\ldots,v_{q+1}\). Label the remaining \(2m+1\) leaves by
    \[
    C_1,a_1,C_2,a_2,\ldots,C_m,a_m,d,
    \]
    where \(d\) is a dummy vertex. For every depth \(r\in\{1,\ldots,q-1\}\), add edges so that the vertices of \(T_L\) at depth \(r\) induce a path.

    \item \textbf{Second Binary Tree:} Construct a mirrored copy \(T_R\), rooted at \(v_p\). Choose one distinguished leaf and label its root-to-leaf path, in reverse order, by \(v_p,v_{p-1},\ldots,v_{q+2}\). Label the remaining leaves by
    \[
    C_1',a_1',C_2',a_2',\ldots,C_m',a_m',d'.
    \]
    Again, at every depth \(r\in\{1,\ldots,q-1\}\), add edges so that the vertices at depth \(r\) induce a path.

    \item \textbf{Connecting Structures:} Add the central edge \((v_{q+1},v_{q+2})\). For every \(j\in[m]\), add the edges \((C_j,C_j')\) and \((a_j,a_j')\), and also add the edge \((d,d')\). On the left side, add the edges \((a_j,a_{j+1})\) for \(j\in[m-1]\), together with \((d,a_1)\) and \((a_1,v_{q+1})\). Symmetrically, add the edges \((a_j',a_{j+1}')\) for \(j\in[m-1]\), together with \((d',a_1')\) and \((a_1',v_{q+2})\).

    \item \textbf{Variable Gadgets:} For every variable \(x_i\), introduce a vertex \(u_i\). For every clause \(C_j\) containing \(x_i\), add the four edges joining \(u_i\) to
    \[
    C_j,\ a_j,\ C_j',\ a_j'.
    \]
\end{enumerate}

The distinguished root-to-leaf paths of the two trees together with the edge \((v_{q+1},v_{q+2})\) induce the path \(v_1,v_2,\ldots,v_p\). The remaining vertices at the same depth will eventually belong to the same witness set as the corresponding vertex of this path.

We briefly explain the purpose of the construction. The vertices \(C_j,a_j\) will belong to the witness set represented by \(v_{q+1}\), while \(C_j',a_j'\) will belong to the witness set represented by \(v_{q+2}\). A variable vertex \(u_i\) can belong only to one of these two central witness sets, corresponding to assigning \(x_i\) the value \true\ or \false. The connectivity of the two central witness sets then enforces the not-all-equal condition in every clause.

We also verify the degree bound. Every internal vertex of either binary tree has degree at most five: one parent, two children, and at most two neighbors on its level. Every clause vertex has degree at most five, and every auxiliary vertex has degree at most eight. The dummy and distinguished leaves have constant degree. Finally, since a variable occurs in at most four clauses, every \(u_i\) has degree at most \(4\cdot4=16\). Hence \(\Delta(G)\leq16\), while \(\Delta(H)\leq2\).

The construction is illustrated in Figure~\ref{fig:maxdegred}. We now prove its correctness.

\begin{figure}[t]
\centering
\begin{tikzpicture}[
    scale=1,
    every node/.style={draw, circle, minimum size=6mm, inner sep=0pt},
    font=\scriptsize
]
\node (v1) at (0,-3.6) {$v_1$};
\node (v2) at (1.5,-1.5) {$v_2$};
\node (v3) at (3,-0.3) {$v_3$};
\node (v4) at (4.5,0) {$v_4$};
\node (v5) at (6.5,0) {$v_5$};
\node (v6) at (8,-0.3) {$v_6$};
\node (v7) at (9.5,-1.5) {$v_7$};
\node (v8) at (11,-3.6) {$v_8$};

\node (d1) at (1.5,-5.7) {};
\node (d2) at (3,-2.7) {};
\node (d3) at (3,-4.7) {};
\node (d4) at (3,-6.7) {};
\node (d1') at (9.5,-5.7) {};
\node (d2') at (8,-2.7) {};
\node (d3') at (8,-4.7) {};
\node (d4') at (8,-6.7) {};
\node (dq) at (4.5,-1) {};
\node (dq') at (6.5,-1) {};
\draw (v1) -- (v2) -- (v3) -- (v4)--(v5) -- (v6) -- (v7) -- (v8);
\draw (v3)--(dq)--(dq')--(v6);

\node (c1)  at ($(v4) + (0,-2.2)$) {$C_1$};
\node (cp1) at ($(v5) + (0,-2.2)$) {$C_1'$};
\node (a1)  at ($(v4) + (0,-3.2)$) {$a_1$};
\node (ap1) at ($(v5) + (0,-3.2)$) {$a_1'$};
\draw (d2) -- (c1) -- (cp1) -- (d2');
\draw (d2) -- (a1) -- (ap1) -- (d2');

\node (c2)  at ($(v4) + (0,-4.2)$) {$C_2$};
\node (cp2) at ($(v5) + (0,-4.2)$) {$C_2'$};
\node (a2)  at ($(v4) + (0,-5.2)$) {$a_2$};
\node (ap2) at ($(v5) + (0,-5.2)$) {$a_2'$};
\draw (d3) -- (c2) -- (cp2) -- (d3');
\draw (d3) -- (a2) -- (ap2) -- (d3');

\node (c3)  at ($(v4) + (0,-6.2)$) {$C_3$};
\node (cp3) at ($(v5) + (0,-6.2)$) {$C_3'$};
\node (a3)  at ($(v4) + (0,-7.2)$) {$a_3$};
\node (ap3) at ($(v5) + (0,-7.2)$) {$a_3'$};
\draw (d4) -- (c3) -- (cp3) -- (d4');
\draw (d4) -- (a3) -- (ap3) -- (d4');

\draw (v1) -- (d1);
\draw (v2) -- (d2);
\draw (d1) -- (d3);
\draw (d1) -- (d4);

\draw (v8) -- (d1');
\draw (v7) -- (d2');
\draw (d1') -- (d3');
\draw (d1') -- (d4');

\draw (v2) -- (d1);
\draw (v3) -- (d2) -- (d3) -- (d4);

\draw (v7) -- (d1');
\draw (v6) -- (d2') -- (d3') -- (d4');

\draw[bend left=50] (a1) to (a2);
\draw[bend left=50] (a2) to (a3);

\draw[bend right=50] (ap1) to (ap2);
\draw[bend right=50] (ap2) to (ap3);

\draw[bend left=43] (ap1) to (v5);
\draw[bend right=43] (a1) to (v4);

\draw (v4) -- (dq);
\draw (v5) -- (dq');

\node (x1) at ($(5.5,-8)$) {$u_i$};

\draw[green!60!black, thick] (x1) -- (a3);
\draw[red!70!black, thick] (x1) -- (ap3);
\draw[green!60!black, thick] (x1) -- (c3);
\draw[red!70!black, thick] (x1) -- (cp3);

\end{tikzpicture}
\caption{Example of the construction of $G$ with $m=3$. The vertices at each
internal level of either binary tree are connected by a path. Dummy leaves are
omitted for readability. We show only one variable vertex $u_i$, corresponding
to a variable appearing in $C_3$.}
\label{fig:maxdegred}
\end{figure}

\noindent $(\Rightarrow)$
Suppose that \(\psi\) admits a satisfying assignment in which every clause contains at least one true and at least one false variable. We construct a witness structure \(\mathcal{W}=\{W_1,\ldots,W_p\}\) of \(G\) into \(H\).

For every \(r\in\{0,\ldots,q-1\}\), put all vertices of \(T_L\) at depth \(r\) into \(W_{r+1}\). Similarly, put all vertices of \(T_R\) at depth \(r\) into \(W_{p-r}\). The horizontal paths added at every internal level ensure that each of these witness sets is connected and contains exactly its corresponding representative.

Put all leaves of \(T_L\) into \(W_{q+1}\), and all leaves of \(T_R\) into \(W_{q+2}\). For every variable \(x_i\), put \(u_i\) into \(W_{q+1}\) if \(x_i=\true\), and into \(W_{q+2}\) if \(x_i=\false\).

We claim that \(W_{q+1}\) is connected. The vertices \(d,a_1,\ldots,a_m,v_{q+1}\) form a connected subgraph. For every clause \(C_j\), there is a variable \(x_i\in C_j\) assigned \true. The vertex \(u_i\in W_{q+1}\) is adjacent to both \(C_j\) and \(a_j\), and hence connects \(C_j\) to this subgraph. Since every variable occurs in at least one clause, every variable vertex assigned \true\ is also connected to this subgraph. Thus \(G[W_{q+1}]\) is connected. The same argument, using a false variable in every clause, shows that \(G[W_{q+2}]\) is connected.

It remains to verify the adjacencies between witness sets. Tree edges join only consecutive levels, while the additional edges inside a level have both endpoints in the same witness set. All edges between leaves of the two trees, as well as all edges incident with the variable vertices, have both endpoints in \(W_{q+1}\cup W_{q+2}\). Hence every edge of \(G\) either has both endpoints in one witness set or joins two consecutive witness sets. Conversely, every pair \(W_i,W_{i+1}\) is adjacent, because of a tree edge or the central edge \((v_{q+1},v_{q+2})\). Therefore,
\[
W_i\text{ and }W_j\text{ are adjacent in }G
\quad\Longleftrightarrow\quad
|i-j|=1.
\]
Thus \(\mathcal{W}\) is a witness structure of \(G\) into \(H\), and \(H\) is a labeled contraction of \(G\).

\noindent $(\Leftarrow)$
Suppose that \(H\) is a labeled contraction of \(G\), and let \(\mathcal{W}=\{W_1,\ldots,W_p\}\) be a corresponding witness structure, where \(v_i\in W_i\).

We first determine the possible witness sets of the vertices in the two binary trees.
For vertices \(u,v\in V(G)\), let \(d_G(u,v)\) denote their distance in \(G\).
If a vertex \(z\) belongs to \(W_i\), then
\[
i-1\leq d_G(v_1,z)
\qquad\text{and}\qquad
p-i\leq d_G(z,v_p),
\]
because consecutive witness sets are the only adjacent witness sets in the target path.

Let \(z\) be a vertex of \(T_L\) at depth \(r\). By construction,
\(d_G(v_1,z)=r\) and \(d_G(z,v_p)=2q+1-r\). Since \(p=2q+2\), the two inequalities imply \(i=r+1\). Thus every vertex of \(T_L\) at depth \(r\) belongs to \(W_{r+1}\). In particular, every leaf of \(T_L\) belongs to \(W_{q+1}\). Symmetrically, every vertex of \(T_R\) at depth \(r\) belongs to \(W_{p-r}\), and every leaf of \(T_R\) belongs to \(W_{q+2}\).

Now consider a variable vertex \(u_i\). Since \(x_i\) occurs in at least one clause, \(u_i\) has a neighbor among the leaves of each tree. Therefore,
\(d_G(v_1,u_i)=q+1\) and \(d_G(u_i,v_p)=q+1\). The same inequalities imply that
\[
u_i\in W_{q+1}\cup W_{q+2}.
\]
We define an assignment by setting \(x_i=\true\) if \(u_i\in W_{q+1}\), and \(x_i=\false\) if \(u_i\in W_{q+2}\).

Consider a clause \(C_j\). Both \(C_j\) and \(a_j\) belong to \(W_{q+1}\). The only neighbors of \(C_j\) that also belong to \(W_{q+1}\) are variable vertices \(u_i\) corresponding to variables \(x_i\in C_j\). Since \(G[W_{q+1}]\) is connected, \(C_j\) must have at least one such neighbor in \(W_{q+1}\). Hence \(C_j\) contains at least one variable assigned \true.

Similarly, \(C_j'\) belongs to \(W_{q+2}\), and its only neighbors inside \(W_{q+2}\) are variable vertices corresponding to variables of \(C_j\). Connectivity of \(G[W_{q+2}]\) therefore implies that at least one variable of \(C_j\) is assigned \false.

Thus every clause contains at least one true and at least one false variable. Hence the assignment satisfies \textsc{Positive-Not-All-$(3,4)$-SAT}, and \(\psi\) is a \yes-instance.

This completes the proof of Theorem~\ref{thm:bounded-max-degree}.

\section{Parameterized by Solution Size Plus Degeneracy}
\label{sec:degeneracy}

We prove Theorem~\ref{thm:bounded-degeneracy-eth} and
Theorem~\ref{thm:degeneracy-algo} in the following two subsections.

\subsection{Conditional Lower Bound}

For the conditional lower bound, we use the following variation of
\textsc{3-SAT}.

\defproblem{\textsc{$1$-in-$3$-SAT}}{A Boolean formula 
$\psi$ in 3-CNF.}{Does there exist a satisfying 
assignment such that exactly one literal is set to true 
in each clause?}

We first observe that the \ETH\ lower bound for \textsc{3-SAT} transfers
to \textsc{$1$-in-$3$-SAT}. Consider a clause
\(C^i=(\ell_1^i\lor \ell_2^i\lor \ell_3^i)\) of a \textsc{3-SAT} instance.
We replace it by the following three clauses, where
\(a_i,b_i,c_i,d_i\) are fresh variables:
\[
(\neg\ell_1^i\lor a_i\lor b_i),\qquad
(\ell_2^i\lor b_i\lor c_i),\qquad
(\neg\ell_3^i\lor c_i\lor d_i).
\]
Each of these clauses is required to contain exactly one true literal.

We briefly verify the correctness of this replacement. If all three literals
\(\ell_1^i,\ell_2^i,\ell_3^i\) are false, then the first clause forces
\(a_i=b_i=\false\), the second forces \(c_i=\true\), and the third contains
both \(c_i\) and \(\neg\ell_3^i\) set to true, which is impossible.
Conversely, suppose that at least one of
\(\ell_1^i,\ell_2^i,\ell_3^i\) is true. If \(\ell_1^i=\true\), set
\(a_i=\ell_2^i\), \(b_i=\neg\ell_2^i\), \(c_i=\false\), and
\(d_i=\ell_3^i\). If \(\ell_1^i=\false\), then the first clause forces
\(a_i=b_i=\false\). If \(\ell_2^i=\true\), set \(c_i=\false\) and
\(d_i=\ell_3^i\); otherwise \(\ell_3^i=\true\), and we set
\(c_i=\true\) and \(d_i=\false\). In each case, all three new clauses have
exactly one true literal.

Thus, each original clause is satisfiable if and only if the three
corresponding clauses admit a $1$-in-$3$ satisfying assignment. The reduction
introduces four new variables and three clauses per original clause, and hence
has only a linear blow-up. Combining this with the \ETH\
~\cite{DBLP:journals/jcss/ImpagliazzoP01} and the Sparsification
Lemma~\cite{DBLP:journals/jcss/ImpagliazzoPZ01}, we obtain the following.

\begin{proposition}
\label{prop:1in3sat-eth}
Unless the \ETH\ fails, \textsc{$1$-in-$3$-SAT} cannot be solved in 
time $2^{o(n + m)}$, where $n$ and $m$ are the number 
of variables and clauses in the input formula, respectively.
\end{proposition}

We now reduce \textsc{$1$-in-$3$-SAT} to \textsc{Labeled Contractibility},
ensuring that the size of the constructed instance is linear in the number of
variables and clauses of the \textsc{$1$-in-$3$-SAT} instance.

We now proceed with the proof of Theorem~\ref{thm:bounded-degeneracy-eth},
which we restate here.

\degeneracyeth*

\begin{proof}
Let $\psi$ be an instance of \textsc{$1$-in-$3$-SAT} with variable set 
\( \{x_1, x_2, \dots, x_n\} \) and clause set 
\( \{C_1, C_2, \dots, C_m\} \). 
We construct a corresponding instance \( (G, H) \) of
\textsc{Labeled Contractibility} such that $\psi$ is satisfiable if and only if
\(H\) is a labeled contraction of \(G\).

We simultaneously define the graphs \(G\) and \(H\) as follows (see
Figure~\ref{fig:combined-gadgets} for an illustration):
\begin{itemize}[nolistsep]
\item \emph{Global vertices:} Add two special vertices 
\(g_T\) and \(g_F\) to both \(G\) and \(H\), connected by 
the edge \((g_T,g_F)\). These represent the `true' and `false'
assignments, respectively, for literals.

\item \emph{Variable gadget:} For each variable 
\(x_i\), \(i\in[n]\):
\begin{itemize}[nolistsep]
    \item Add a vertex \(u_i\) to both \(G\) and \(H\), representing variable \(x_i\).
    \item Add two vertices \(v_i\) and \(v'_i\) to \(G\), corresponding to literals \(x_i\) and \(\neg x_i\), respectively.
    \item In \(G\), connect \(u_i\) to both \(v_i\) and \(v'_i\). Also connect both \(v_i\) and \(v'_i\) to \(g_T\) and \(g_F\).
    \item In \(H\), add edges \((u_i,g_T)\), \((u_i,g_F)\), and retain \((g_T,g_F)\).
\end{itemize}

\item \emph{Clause gadget:} 
For each clause \(C_j=(\ell_{j1}\lor\ell_{j2}\lor\ell_{j3})\), 
let \(\mathsf{litv}(\ell)\) denote \(v_i\) if \(\ell=x_i\), and 
\(v_i'\) if \(\ell=\neg x_i\).
\begin{itemize}[nolistsep]
    \item Add four vertices \(w_{j0},w_{j1},w_{j2},w_{j3}\) to both \(G\) and \(H\).
    \item In \(H\), connect \(w_{j0}\) to \(g_T\) and \(g_F\).
    \item In \(G\), connect \(w_{j0}\) to
    \(\mathsf{litv}(\ell_{j1})\),
    \(\mathsf{litv}(\ell_{j2})\), and
    \(\mathsf{litv}(\ell_{j3})\).
    \item For each \(k\in\{1,2,3\}\), add edge \((w_{jk},g_T)\) to both
    \(G\) and \(H\), and edge \((w_{jk},g_F)\) only to \(H\).
    \item Finally, form a 6-cycle in \(G\) by connecting the three
    \(w_{jk}\)'s with the three
    \(\mathsf{litv}(\ell_{jk})\)'s in an alternating cycle such that
    no two \(w_{jk}\)'s nor any two
    \(\mathsf{litv}(\ell_{jk})\)'s are adjacent.
\end{itemize}
\end{itemize}

This completes the reduction. We first observe that both constructed graphs
have constant degeneracy. In \(G\), first remove all clause vertices
\(w_{j0},w_{j1},w_{j2},w_{j3}\); each has degree at most three at this point.
Next remove all vertices \(u_i\), each of which has degree two. After these
vertices have been removed, every literal vertex \(v_i,v_i'\) has only
\(g_T\) and \(g_F\) as possible remaining neighbors. Thus \(G\) is
3-degenerate. In \(H\), every vertex other than \(g_T,g_F\) has degree two,
and hence \(H\) is 2-degenerate.

\begin{figure}[t]
\centering

\begin{subfigure}[t]{0.45\textwidth}
\centering
\begin{tikzpicture}[scale=0.75, every node/.style={draw, circle, minimum size=6mm, inner sep=0pt}, font=\scriptsize]
\node (gt) at (0,0) {$g_T$};
\node (vi) at (2,0) {$v_i$};
\node (ui) at (4,0) {$u_i$};
\node (gf) at (0,-1) {$g_F$};
\node (vpi) at (2,-1) {$v_i'$};
\draw (gt) -- (gf);
\draw (vi) -- (ui);
\draw (vpi) -- (ui);
\draw[green!60!black, thick] (gt) -- (vi);
\draw[green!60!black, thick] (gt) -- (vpi);
\draw[red!70!black, thick] (gf) -- (vi);
\draw[red!70!black, thick] (gf) -- (vpi);
\end{tikzpicture}
\caption{Variable gadget in $G$}
\label{fig:variable-gadget-G}
\end{subfigure}
\hfill
\begin{subfigure}[t]{0.45\textwidth}
\centering
\begin{tikzpicture}[scale=0.75, every node/.style={draw, circle, minimum size=6mm, inner sep=0pt}, font=\scriptsize]
\node (gt) at (0,0) {$g_T$};
\node (ui) at (2,0) {$u_i$};
\node (gf) at (0,-1) {$g_F$};
\draw (gt) -- (gf);
\draw[blue!70!black, thick] (gt) -- (ui);
\draw[blue!70!black, thick] (gf) -- (ui);
\end{tikzpicture}
\caption{Variable gadget in $H$}
\label{fig:variable-gadget-H}
\end{subfigure}

\vspace{0.75em}

\begin{subfigure}[t]{0.45\textwidth}
\centering
\begin{tikzpicture}[scale=0.8, every node/.style={draw, circle, minimum size=6mm, inner sep=0pt}, font=\scriptsize]
\node (gt) at (0,0) {$g_T$};
\node (vi) at (2,0) {$v_i$};
\node (vj) at (2,-1) {$v_j$};
\node (vk') at (2,-2) {$v_k'$};
\node (wo) at (-1,-1) {$w_0$};
\node (w1) at (4,0) {$w_1$};
\node (w2) at (4,-1) {$w_2$};
\node (w3) at (4,-2) {$w_3$};
\node (gf) at (0,-2) {$g_F$};
\draw (gt) -- (gf);

\draw[green!60!black, thick] (gt) -- (vi);
\draw[green!60!black, thick] (gt) -- (vj);
\draw[green!60!black, thick] (gt) -- (vk');
\draw[red!70!black, thick] (gf) -- (vi);
\draw[red!70!black, thick] (gf) -- (vj);
\draw[red!70!black, thick] (gf) -- (vk');
\draw[teal!70!black, thick] (wo) -- (vi);
\draw[teal!70!black, thick] (wo) -- (vj);
\draw[teal!70!black, thick] (wo) -- (vk');
\draw[yellow!40!red, thick] (w1) -- (vi);
\draw[yellow!40!red, thick] (w2) -- (vi);
\draw[yellow!40!red, thick] (w2) -- (vj);
\draw[yellow!40!red, thick] (w3) -- (vk');
\draw[yellow!40!red, thick] (w1) -- (vk');
\draw[yellow!40!red, thick] (w3) -- (vj);

\draw[blue!70!black, thick] (gt) to[out=30,in=150] (w1);
\draw[blue!70!black, thick] (gt) to[out=0,in=180] (w2);
\draw[blue!70!black, thick] (gt) to[out=-45,in=170] (w3);

\end{tikzpicture}
\caption{Clause gadget in $G$}
\label{fig:clause-gadget-G}
\end{subfigure}
\hfill
\begin{subfigure}[t]{0.45\textwidth}
\centering
\begin{tikzpicture}[scale=1, every node/.style={draw, circle, minimum size=6mm, inner sep=0pt}, font=\scriptsize]
\node (gt) at (2,0) {$g_T$};
\node (wo) at (0,-1) {$w_0$};
\node (w1) at (4,0) {$w_1$};
\node (w2) at (4,-1) {$w_2$};
\node (w3) at (4,-2) {$w_3$};
\node (gf) at (2,-2) {$g_F$};
\draw (gt) -- (gf);
\draw[blue!70!black, thick] (gt) -- (wo);
\draw[blue!70!black, thick] (gt) -- (w1);
\draw[blue!70!black, thick] (gt) -- (w2);
\draw[blue!70!black, thick] (gt) -- (w3);
\draw[blue!70!black, thick] (gf) -- (wo);
\draw[blue!70!black, thick] (gf) -- (w1);
\draw[blue!70!black, thick] (gf) -- (w2);
\draw[blue!70!black, thick] (gf) -- (w3);
\end{tikzpicture}
\caption{Clause gadget in $H$}
\label{fig:clause-gadget-H}
\end{subfigure}

\caption{Illustration of the reduction. (Top) Variable gadgets in $G$ and $H$
for a variable $x_i$. (Bottom) Clause gadgets in $G$ and $H$ for a clause
$(x_i\lor x_j\lor \neg x_k)$, with corresponding literal vertices
$v_i,v_j,v_k'$.}
\label{fig:combined-gadgets}
\end{figure}

We now show the equivalence between satisfiability of \(\psi\) and labeled
contractibility of \(G\) into \(H\).

\noindent
(\(\Rightarrow\)) Suppose there exists a satisfying assignment
\(\alpha:\{x_1,\dots,x_n\}\to\{\true,\false\}\) such that exactly one
literal is true in each clause. We construct a witness structure for \(G\)
into \(H\).

Let \(W(g_T)\) contain \(g_T\) together with every literal vertex whose
corresponding literal is true under \(\alpha\), and let \(W(g_F)\) contain
\(g_F\) together with every literal vertex whose corresponding literal is
false under \(\alpha\). For every other vertex \(h\in V(H)\), let
\(W(h)=\{h\}\).

Both \(W(g_T)\) and \(W(g_F)\) are connected, since every literal vertex is
adjacent to both \(g_T\) and \(g_F\) in \(G\). All other witness sets are
singletons.

For every variable \(x_i\), exactly one of \(v_i,v_i'\) represents a true
literal and the other represents a false literal. Since \(u_i\) is adjacent
to both \(v_i\) and \(v_i'\), the witness set \(W(u_i)\) is adjacent to both
\(W(g_T)\) and \(W(g_F)\), as required in \(H\).

Now consider a clause
\(C_j=(\ell_{j1}\lor\ell_{j2}\lor\ell_{j3})\). Since exactly one of its
literals is true, exactly one of the three corresponding literal vertices
belongs to \(W(g_T)\), while the other two belong to \(W(g_F)\).
Consequently, \(w_{j0}\), which is adjacent to all three literal vertices,
is adjacent to both global witness sets. Moreover, each \(w_{jk}\) is already
adjacent to \(g_T\), and its two neighbors on the alternating 6-cycle cannot
both correspond to the unique true literal. Hence at least one of these two
neighbors belongs to \(W(g_F)\), and therefore \(W(w_{jk})\) is also adjacent
to \(W(g_F)\).

These are precisely the adjacencies required in \(H\). Furthermore, every
edge of \(G\) either lies inside one of the two global witness sets or joins
one of them to a vertex that is adjacent to the corresponding global vertex
in \(H\). Hence no adjacency corresponding to a non-edge of \(H\) is created.
Thus the sets above form a valid witness structure, and \(H\) is a labeled
contraction of \(G\).

\noindent
(\(\Leftarrow\)) Conversely, suppose \(H\) is a labeled contraction of \(G\),
and let \(\mathcal{W}\) be a corresponding witness structure.

We first observe that every literal vertex that occurs in a clause must belong
to either \(W(g_T)\) or \(W(g_F)\). Indeed, literal vertices are pairwise
non-adjacent. Hence, if a literal vertex belongs to a witness set represented
by some vertex \(h\in V(H)\), connectivity implies that it must be adjacent
to \(h\) in \(G\). If it belonged to \(W(u_i)\), then its edge to a clause
vertex in which the literal occurs would create an adjacency between
\(u_i\) and that clause vertex, which is absent from \(H\). On the other hand,
if it belonged to a witness set represented by a clause vertex, then its edge
to \(u_i\) would create a forbidden adjacency between that clause vertex and
\(u_i\). Thus an occurring literal vertex can belong only to \(W(g_T)\) or
\(W(g_F)\).

We now define a truth assignment. If \(v_i\) occurs in the formula, set
\(\alpha(x_i)=\true\) if and only if \(v_i\in W(g_T)\). If \(v_i\) does not
occur but \(v_i'\) does, choose \(\alpha(x_i)\) so that the literal
\(\neg x_i\) is true if and only if \(v_i'\in W(g_T)\). Variables that do
not occur in the formula may be assigned arbitrarily.

This definition is consistent when both \(v_i\) and \(v_i'\) occur.
In this case both literal vertices belong to the global witness sets and
\(W(u_i)=\{u_i\}\). Since \(u_i\) is adjacent in \(H\) to both \(g_T\) and
\(g_F\), while its only neighbors in \(G\) are \(v_i\) and \(v_i'\), one of
these vertices must belong to \(W(g_T)\) and the other to \(W(g_F)\).
Therefore, for every literal occurring in the formula, the literal is true
under \(\alpha\) if and only if its corresponding literal vertex belongs to
\(W(g_T)\).

Consider now a clause
\(C_j=(\ell_{j1}\lor\ell_{j2}\lor\ell_{j3})\).
The witness set of each clause vertex \(w_{jk}\) is a singleton: an occurring
literal vertex cannot belong to such a witness set by the argument above,
and no other non-representative vertex is adjacent to it.

Since \(w_{j0}\) is adjacent to \(g_T\) in \(H\), while in \(G\) its only
neighbors are the three literal vertices corresponding to the clause, at
least one of these literal vertices belongs to \(W(g_T)\). Thus at least one
literal of \(C_j\) is true.

For each \(k\in\{1,2,3\}\), the vertex \(w_{jk}\) must also be adjacent to
\(g_F\) in \(H\). However, \((w_{jk},g_F)\notin E(G)\); its only other
neighbors are the two literal vertices adjacent to it on the alternating
6-cycle. Hence at least one of those two literal vertices must belong to
\(W(g_F)\). The three vertices \(w_{j1},w_{j2},w_{j3}\) correspond to the
three pairs of clause literals. Therefore, no two clause literals can both
belong to \(W(g_T)\). Combining this with the previous paragraph, exactly
one literal of \(C_j\) is true.

Thus \(\alpha\) is a satisfying assignment of \(\psi\) in which exactly one
literal is true in every clause.

Finally, the construction contains \(\mathcal{O}(n+m)\) vertices and edges
and can be produced in polynomial time. Therefore, if
\textsc{Labeled Contractibility} could be solved in time
\(2^{o(|V(G)|+|E(G)|)}\) on graphs of bounded degeneracy, then
\textsc{$1$-in-$3$-SAT} could be solved in time \(2^{o(n+m)}\),
contradicting Proposition~\ref{prop:1in3sat-eth}. This completes the proof
of Theorem~\ref{thm:bounded-degeneracy-eth}.
\end{proof}

\subsection{Algorithmic Result}

We first present the proof of Theorem~\ref{thm:degeneracy-algo},
establishing an alternative \FPT\ algorithm, and later mention its consequences.

\degeneracy*

\begin{proof}
Let \( (G, H) \) be an instance of the \textsc{Labeled Contraction} 
problem. After a sanity check, we may assume that \( V(H) \subseteq V(G) \).

The goal is to decide whether there exists a valid witness structure 
\(\mathcal{W} = \{W_h \subseteq V(G) \mid h \in V(H)\}\) such that contracting each 
\( W_h \) into the vertex \( h \) yields the graph \( H \).

Let \( k = |V(G) \setminus V(H)| \) denote the number of vertices in 
\( G \) that are not in \( H \), and hence must be contracted.

Let \( \delta(H) \) denote the degeneracy of \( H \). 
By definition, \( H \) admits an ordering \( v_1, v_2, \ldots, v_{|V(H)|} \) 
such that each vertex has at most \( \delta(H) \) neighbors among its predecessors.
This implies that \( H \) is \( (\delta(H)+1) \)-colorable, and such a coloring can
be computed in polynomial time.

Fix a proper coloring
\( c: V(H) \to [\delta(H)+1] \).
For every vertex \(x \in V(G)\setminus V(H)\), we guess the color of the
representative of the witness set containing \(x\). Thus, we extend \(c\) to
a mapping
\(c':V(G)\to[\delta(H)+1]\)
such that \(c'(v)=c(v)\) for every \(v\in V(H)\).
There are at most
\((\delta(H)+1)^k\)
such extensions.

For a fixed extension \(c'\), consider the connected components of the
subgraphs of \(G\) induced by individual color classes. We claim that
\(c'\) corresponds to a valid witness structure if and only if the following
conditions hold:
\begin{enumerate}[nolistsep]
    \item every monochromatic connected component contains exactly one
    vertex of \(V(H)\); and
    \item after contracting every monochromatic connected component to its
    unique vertex of \(V(H)\), the resulting labeled graph is exactly \(H\).
\end{enumerate}

To see the forward direction, suppose that \(\mathcal{W}\) is a valid witness
structure and each vertex is assigned the color of the representative of its
witness set. Then every witness set is monochromatic and connected.
Furthermore, two distinct witness sets with representatives of the same color
cannot be adjacent in \(G\), since their representatives are non-adjacent in
the proper coloring of \(H\). Consequently, every monochromatic connected
component is exactly one witness set, and hence contains exactly one vertex
of \(V(H)\). Contracting these components therefore yields \(H\).

Conversely, if the two conditions above hold, the monochromatic connected
components form a partition of \(V(G)\) into connected sets, each containing
exactly one vertex of \(V(H)\), and their adjacency relation is exactly that
of \(H\). Thus they form a valid witness structure of \(G\) into \(H\).

For each extension \(c'\), the monochromatic connected components and their
adjacency relation can be computed and checked in polynomial time. Therefore,
the total running time is
\[
(\delta(H)+1)^k \cdot |V(G)|^{\mathcal{O}(1)}.
\]
This completes the proof of Theorem~\ref{thm:degeneracy-algo}.
\end{proof}

The degeneracy coloring in the proof can be replaced by an optimal proper
coloring to obtain the following corollary.

\begin{corollary}
The \textsc{Labeled Contraction} problem admits
an algorithm with running time
\(2^{|V(H)|}\cdot |V(H)|^{\mathcal{O}(1)}
+ \chi(H)^k \cdot |V(G)|^{\mathcal{O}(1)}\),
where \(\chi(H)\) is the chromatic number of \(H\).
\end{corollary}

\begin{proof}
Let \(\chi(H)\) denote the chromatic number of \(H\), and fix an optimal
proper coloring of \(H\) with \(\chi(H)\) colors.
For each vertex in \(V(G)\setminus V(H)\), we guess the color of the
representative of its witness set. Hence there are at most
\(\chi(H)^k\) extensions of the coloring to \(V(G)\).

For every such extension, we compute the monochromatic connected components
and verify, as in the proof of Theorem~\ref{thm:degeneracy-algo}, that each
component contains exactly one vertex of \(V(H)\) and that contracting these
components yields exactly \(H\). This takes polynomial time for each extension.

An optimal coloring of \(H\) can be computed in time
\(2^{|V(H)|}\cdot |V(H)|^{\mathcal{O}(1)}\) using the exact coloring algorithm in ~\cite{DBLP:journals/siamcomp/BjorklundHK09}.
Hence the claimed running time follows.
\end{proof}

\section{Brute-force Algorithm and its Optimality}
\label{sec:brute-force}

We now prove Theorem~\ref{thm:brute-force-optimal}. 
We restate it for the reader's convenience.

\Bruteforce*

\subsection{Brute-force Algorithm}

We prove the following lemma.

\begin{lemma}
\label{thm:brute-force-algo-max-common-label-contraction}
The \textsc{Maximum Common Labeled Contraction} problem admits an algorithm with running time $|V(H)|^{\calO(|V(G)|)}$.
\end{lemma}

\begin{proof}
We begin by describing an algorithm for the 
\textsc{Labeled Contractibility} problem with running time 
$|V(H)|^{\calO(|V(G)|)}$, which we later use as a subroutine to solve 
\textsc{Maximum Common Labeled Contraction}.

As established in Observation~2 of~\cite{DBLP:journals/corr/abs-2502-16096}, 
a graph $H$ is a labeled contraction of a graph $G$ if and only 
if there exists a witness structure $W = \{W_1, \ldots, W_{|V(H)|}\}$ of $G$ into $H$.
Recall that a witness structure is a partition of $V(G)$ into $|V(H)|$ non-empty sets, 
where each $W_i$ corresponds to a distinct vertex of $H$. 

The algorithm exhaustively enumerates all such candidate partitions and 
checks their validity. 
Formally, it distributes all the vertices in $V(G) \setminus V(H)$ 
into the $|V(H)|$ parts. 
Since each of the $|V(G)| - |V(H)|$ vertices can be assigned to any of the 
$|V(H)|$ parts, the number of such candidate partitions is at most 
$(|V(H)|)^{|V(G)| - |V(H)|}$.

For each candidate witness structure $W$, the algorithm performs the following checks:
$(i)$ For each $W_j \in W$, verify that the induced subgraph $G[W_j]$ is connected.
$(ii)$ For all distinct $u, v \in V(H)$, verify that $uv \in E(H)$ if and only if there exists at least one edge in $G$ between some vertex in $W(u)$ and some vertex in $W(v)$. 
Both these steps take time which is polynomial in $|V(G)|$.
If a partition satisfies both conditions, then $W$ is a valid witness structure certifying 
that $H$ is a labeled contraction of $G$. 
If no such structure is found, the algorithm correctly 
concludes that $H$ is not a labeled contraction of $G$.
Thus, the total running time of this procedure is bounded by
$(|V(H)|)^{|V(G)| - |V(H)|} \cdot \calO(|V(G)|^2) = |V(H)|^{\calO(|V(G)|)}$.

We now describe the algorithm for \textsc{Maximum Common Labeled Contraction}. 
Let $G$ and $H$ be the two input graphs. Since the problem is symmetric in
$G$ and $H$, we assume without loss of generality that
$|V(G)| \ge |V(H)|$.
The algorithm enumerates all possible candidate graphs $H'$ that can be obtained 
as a labeled contraction of $H$, and for each candidate, 
checks whether it is also a labeled contraction of $G$.

To bound the number of such candidate graphs $H'$, we count the number of 
candidate witness structures from $H$ to $H'$. 
For a fixed size $n' = |V(H')|$, 
there are $\binom{|V(H)|}{n'}$ ways to choose the 
representative vertices of $H'$, and the remaining $|V(H)| - n'$ 
vertices can be assigned to these parts in $(n')^{|V(H)| - n'}$ ways. 
Summing over all possible values of $n'$, the total number of such candidates is bounded by
\[
\sum_{n'=1}^{|V(H)|} \binom{|V(H)|}{n'} \cdot (n')^{|V(H)| - n'}
\le |V(H)|^{|V(H)|} \cdot 2^{|V(H)|}
= |V(H)|^{\calO(|V(H)|)}.
\]

For each such candidate, we verify in polynomial time whether it defines a
valid witness structure of $H$ and hence determines a labeled contraction $H'$.
We then check whether $H'$ is also a labeled contraction of $G$, using the
subroutine above in time
$|V(H')|^{\calO(|V(G)|)} \le |V(H)|^{\calO(|V(G)|)}$.
Finally, since obtaining $H'$ from $G$ and $H$ requires exactly
$|V(G)| + |V(H)| - 2|V(H')|$ contractions, we accept if this quantity is at most $k$.
Hence, the overall running time of the algorithm is
$|V(H)|^{\mathcal{O}(|V(G)|)}$.
This concludes the proof of the lemma.
\end{proof}

\subsection{Conditional Lower Bound}

In this subsection, 
we present the conditional lower bound, demonstrating that the algorithm is 
essentially optimal unless the Exponential Time Hypothesis (\ETH) fails.

Our lower bound is established through a polynomial-time reduction from the \textsc{Cross Matching} problem, introduced by Fomin et al.~\cite{DBLP:journals/toct/FominLMSZ21}. This problem has been a hardness source for several other contraction-based graph problems.

\defproblem{\textsc{Cross Matching}}{A graph $G'$ together with a partition $(A, B)$ of $V(G')$ such that $|A| = |B|$.}{Does there exist a perfect matching $M \subseteq E(G')$ such that each edge in $M$ connects a vertex from $A$ to one from $B$, and the graph $G'/M$ is a clique?}

\begin{proposition}[Lemma 4.1 in \cite{DBLP:journals/toct/FominLMSZ21}]
\label{prop:cross-matching-lower-bound}
Unless the \ETH\ fails, \textsc{Cross Matching} cannot be solved in time $n^{o(n)}$, where $n = |A| = |B|$.
\end{proposition}

We now present a polynomial-time reduction from \textsc{Cross Matching} to \textsc{Labeled Contractibility}.

\begin{figure}[t]
\centering

\begin{subfigure}[t]{0.48\textwidth}
\centering
\begin{tikzpicture}[scale=1,
    dot/.style={circle, fill, inner sep=1pt},
    edge/.style={thick, purple},
    vc/.style={orange, thick},
    subs/.style={violet, thick},
    new/.style={olive, thick},
    side/.style={gray, thick},
    every node/.style={font=\small}
]

\node[ellipse, draw, minimum height=3cm, minimum width=1.2cm] (A) at (0,2.5) {};
\node at ($(A)+(-1,1)$) {$A$};

\node[dot, label=above:$a_p$] (a4) at ($(A)+(0,-1)$) {};
\node at ($(A)+(0,-0.15)$) {$\vdots$};
\node[dot, label=above:$a_2$] (a2) at ($(A)+(0,0.2)$) {};
\node[dot, label=above:$a_1$] (a1) at ($(A)+(0,0.8)$) {};

\node[ellipse, draw, minimum height=7.5mm, minimum width=7.5mm] (x1) at (5,3) {$x_1$};
\node[ellipse, draw, minimum height=7.5mm, minimum width=7.5mm] (x2) at (5,2) {$x_2$};

\node[ellipse, draw, minimum height=3cm, minimum width=1.2cm] (B) at (3,2.5) {};
\node at ($(B)+(1,1)$) {$B$};

\node[dot, label=above:$b_q$] (b4) at ($(B)+(0,-1)$) {};
\node at ($(B)+(0,-0.15)$) {$\vdots$};
\node[dot, label=above:$b_2$] (b2) at ($(B)+(0,0.2)$) {};
\node[dot, label=above:$b_1$] (b1) at ($(B)+(0,0.8)$) {};

\draw[edge] (x1) -- (b1);
\draw[edge] (x1) -- (b2);
\draw[edge] (x1) -- (b4);
\draw[edge] (x2) -- (b1);
\draw[edge] (x2) -- (b2);
\draw[edge] (x2) -- (b4);

\draw[side] (a1) -- (b4);
\draw[side] (a2) to[bend left = 60] (a1);
\draw[side] (a4) -- (b4);
\draw[side] (a2) -- (b2);
\draw[side] (a1) -- (b2);
\draw[side] (a4) -- (b1);

\end{tikzpicture}
\caption{Graph $G$}
\end{subfigure}
\hfill
\begin{subfigure}[t]{0.48\textwidth}
\centering
\begin{tikzpicture}[scale=1,
    dot/.style={circle, fill, inner sep=1pt},
    edge/.style={thick, purple},
    vc/.style={orange, thick},
    subs/.style={violet, thick},
    new/.style={olive, thick},
    side/.style={darkgray, thick},
    every node/.style={font=\small}
]

\node[draw, minimum height=2.7cm, minimum width=1.2cm] (A) at (0,2.5) {};
\node at ($(A)+(-1,0)$) {$A$};

\node[dot, label=above:$a_p$] (a4) at ($(A)+(0.15,-1)$) {};
\node at ($(A)+(0,-0.15)$) {$\vdots$};
\node[dot, label=above:$a_2$] (a2) at ($(A)+(0.15,0.2)$) {};
\node[dot, label=above:$a_1$] (a1) at ($(A)+(0.15,0.8)$) {};

\node[ellipse, draw, minimum height=7.5mm, minimum width=7.5mm] (x1) at (2,3) {$x_1$};
\node[ellipse, draw, minimum height=7.5mm, minimum width=7.5mm] (x2) at (2,2) {$x_2$};

\draw[new] (x1) -- (a1);
\draw[new] (x1) -- (a2);
\draw[new] (x1) -- (a4);
\draw[new] (x2) -- (a1);
\draw[new] (x2) -- (a2);
\draw[new] (x2) -- (a4);

\draw[subs] (a1) to[bend right=60] (a2);
\draw[subs] (a1) to[bend right=60] (a4);
\draw[subs] (a2) to[bend right=60] (a4);

\end{tikzpicture}
\caption{Graph $H$}
\end{subfigure}

\caption{Illustration of the graphs $G$ and $H$ in the reduction.
The required adjacency across $\{x_1, x_2\}$ and $A$ forces 
a candidate solution to form a matching across $A$ and $B$.}
\label{fig:GH-sidebyside}
\end{figure}

\begin{lemma}
\label{lemma:lower-bound-label-contraction}
Unless the \ETH\ fails, the \textsc{Labeled Contractibility} problem does not admit an algorithm running in time $|V(H)|^{o(|V(G)|)}$.
\end{lemma}

\begin{proof}
Let $(G', A, B)$ be an instance of \textsc{Cross Matching} with $|A| = |B| = n$. We construct an equivalent instance $(G, H)$ of \textsc{Labeled Contractibility} as follows.

The graph $G$ is obtained from $G'$ by introducing two new vertices, $x_1$ and $x_2$, and adding edges from each to all vertices in $B$. Formally, $V(G) = V(G') \cup \{x_1, x_2\}$ and $E(G) = E(G') \cup \{(x_1, b), (x_2, b) \mid b \in B\}$.

The target graph $H$ is defined on the vertex set $V(H) = A \cup \{x_1, x_2\}$, with the following edge set:
(i) all pairs $\{a_1, a_2\} \subseteq A$ are adjacent (i.e., $A$ induces a clique); and
(ii) both $x_1$ and $x_2$ are adjacent to every vertex in $A$.
Vertices $x_1$ and $x_2$ are not adjacent. Equivalently, $H$ is a complete graph on $A \cup \{x_1, x_2\}$ minus the edge $(x_1, x_2)$. See Figure~\ref{fig:GH-sidebyside} for an illustration.

We prove that $(G', A, B)$ is a \yes-instance of \textsc{Cross Matching} if and only if $(G, H)$ is a \yes-instance of \textsc{Labeled Contractibility}.

{($\Rightarrow$)} Assume $(G', A, B)$ is a \yes-instance. Then there exists a perfect matching $M$ between $A$ and $B$ such that the contracted graph $G'/M$ is a clique. For every $a\in A$, let $b_a\in B$ denote its unique mate in $M$.

Define witness sets $W(a)=\{a,b_a\}$ for every $a\in A$, and
$W(x_1)=\{x_1\}$, $W(x_2)=\{x_2\}$.
Since $(a,b_a)\in E(G')$, every witness set is connected.
Moreover, since $G'/M$ is a clique, every two distinct witness sets
$W(a)$ and $W(a')$ are adjacent in $G$.
Since every $b_a$ is adjacent to both $x_1$ and $x_2$, each of
$W(x_1)$ and $W(x_2)$ is adjacent to every $W(a)$.
Finally, $x_1$ and $x_2$ are non-adjacent in $G$.
Hence the adjacency relation among the witness sets is exactly that of $H$.
Therefore, $H$ is a labeled contraction of $G$.

{($\Leftarrow$)} Assume that $H$ is a labeled contraction of $G$, and let
$\mathcal{W}$ be a corresponding witness structure.

We first show that no vertex of $B$ can belong to $W(x_1)$ or $W(x_2)$.
Suppose that some $b\in B$ belongs to $W(x_1)$. Since $b$ is adjacent to
$x_2$ in $G$, the witness sets $W(x_1)$ and $W(x_2)$ would be adjacent,
contradicting $(x_1,x_2)\notin E(H)$. The same argument applies symmetrically
to $W(x_2)$. Thus $W(x_1)=\{x_1\}$ and $W(x_2)=\{x_2\}$, and every
$b\in B$ belongs to a witness set represented by some vertex of $A$.

Now fix $a\in A$. Since $(a,x_1)\in E(H)$, the witness sets $W(a)$ and
$W(x_1)$ must be adjacent in $G$. As $W(x_1)=\{x_1\}$ and the neighbors
of $x_1$ in $G$ are exactly the vertices of $B$, the set $W(a)$ must contain
at least one vertex of $B$.

There are $|A|=|B|=n$ witness sets represented by vertices of $A$, and the
$n$ vertices of $B$ are distributed among them, with each such witness set
containing at least one vertex of $B$. Hence every $W(a)$ contains exactly
one vertex $b_a\in B$. Since $W(a)=\{a,b_a\}$ is connected, we have
$(a,b_a)\in E(G')$. Therefore,
$M=\{(a,b_a)\mid a\in A\}$ is a perfect matching between $A$ and $B$.

Finally, because $A$ induces a clique in $H$, every two distinct witness sets
$W(a)$ and $W(a')$ are adjacent in $G$. After contracting every edge
$(a,b_a)\in M$, this implies that the vertices of $A$ form a clique.
Hence $G'/M$ is a clique, and $(G',A,B)$ is a \yes-instance of
\textsc{Cross Matching}.

The construction of $(G, H)$ from $(G', A, B)$ takes polynomial time. The sizes of the graphs are $|V(G)| = 2n + 2$ and $|V(H)| = n + 2$. Therefore, if \textsc{Labeled Contractibility} could be solved in time $|V(H)|^{o(|V(G)|)}$, then \textsc{Cross Matching} could be solved in time $n^{o(n)}$, which contradicts Proposition~\ref{prop:cross-matching-lower-bound}. This completes the proof of the lemma.
\end{proof}
\section{Conclusion}
\label{sec:conclusion}

In this article, we advanced the study of the parameterized complexity of \textsc{Labeled Contractibility}, building upon the foundational work of Lafond and Marchand~\cite{DBLP:journals/corr/abs-2502-16096}. We developed a constructive dynamic programming algorithm with a running time of \( 2^{\mathcal{O}(\tw^2)} \cdot n^{\mathcal{O}(1)} \), which significantly improves upon the previously known meta-theorem-based approach with non-elementary dependence on treewidth. Furthermore, we complemented this result with a matching conditional lower bound under the Exponential Time Hypothesis (\ETH), which is a relatively rare phenomenon.
Our work resolves several open questions from~\cite{DBLP:journals/corr/abs-2502-16096}. 
We provided an improved \FPT\ algorithm with respect to 
the combined parameter of the number of contractions and degeneracy. 
We also established that the straightforward brute-force approach for 
the problem is optimal under \ETH.

We conclude by highlighting some directions for future research. 
We conjecture that our dynamic programming framework over 
tree decompositions could be extended to solve the more 
general \textsc{Maximum Common Labeled Contraction} problem.
Additionally, it would be interesting to explore whether 
fixed-parameter algorithms with better running time
can be obtained when parameterized by
the feedback vertex set number or the vertex cover number.

Another natural direction is to study \textsc{Labeled Contractibility} when
the target graph is restricted to simple graph classes, such as paths or cycles
of fixed order. It would be interesting to understand how the prescribed labels
affect the complexity of these classical contraction problems, and more generally
to investigate the problem for other fixed target graphs.

\bibliographystyle{plain}
\bibliography{references}

\end{document}